\documentclass[11pt]{article}

\newif\ifnotes
\notesfalse

\usepackage{hyperref}
\usepackage{fullpage}
\usepackage{amssymb}
\usepackage{amsmath}
\usepackage{amsthm}
\usepackage{amsfonts}
\usepackage{bm}
\usepackage{enumitem}
\usepackage{color}
\usepackage{comment}
\usepackage[capitalize]{cleveref}
\usepackage[dvipsnames]{xcolor}
\usepackage{float}
\usepackage[T1]{fontenc}
\usepackage{todonotes}
\usepackage{asymptote}
\usepackage{mdframed}
\usepackage[most]{tcolorbox}
\usepackage{hyperref}
\usepackage{enumitem}
\usepackage{framed}
\usepackage{mdframed}
\usepackage{scrextend}
\usepackage{bbm}
\usepackage{thm-restate}
\usepackage{xcolor}
\usepackage{algorithm}
\usepackage[noend]{algpseudocode}
\usepackage{float}
\usepackage{tikz}
\usetikzlibrary{positioning, fit, decorations.pathreplacing, arrows.meta, calc}

\tikzset{
    ncbar angle/.initial=90,
    ncbar/.style={
        to path=(\tikztostart)
        -- ($(\tikztostart)!#1!\pgfkeysvalueof{/tikz/ncbar angle}:(\tikztotarget)$)
        -- ($(\tikztotarget)!($(\tikztostart)!#1!\pgfkeysvalueof{/tikz/ncbar angle}:(\tikztotarget)$)!\pgfkeysvalueof{/tikz/ncbar angle}:(\tikztostart)$)
        -- (\tikztotarget)
    },
    ncbar/.default=0.5cm,
}

\tikzset{square left brace/.style={ncbar=2mm}}
\tikzset{square right brace/.style={ncbar=-2mm}}

\algrenewcommand\textproc{}

\newcommand{\mnote}[1]{\ifnotes $\ll$\textsf{\color{magenta} Meghal: { #1}}$\gg$ \fi}
\notesfalse

\usepackage[T1]{fontenc}

\definecolor{denim}{rgb}{0.08, 0.38, 0.74}
\definecolor{periwinkle}{rgb}{0.6, 0.6, 0.95}
\definecolor{wildblueyonder}{rgb}{0.64, 0.68, 0.82}

\usepackage{hyperref}
\hypersetup{
    colorlinks=true,
    linkcolor=denim,
    filecolor=denim,
    citecolor=periwinkle,
    urlcolor=periwinkle,
}
\usepackage[hyperpageref]{backref}

\newtheorem{theorem}{Theorem}[section]

\newtheorem{definition}[theorem]{Definition}
\newtheorem{lemma}[theorem]{Lemma}

\theoremstyle{remark}

\Crefname{theorem}{Theorem}{Theorems}
\Crefname{claim}{Claim}{Claims}
\Crefname{lemma}{Lemma}{Lemmas}
\Crefname{proposition}{Proposition}{Propositions}
\Crefname{corollary}{Corollary}{Corollaries}
\Crefname{definition}{Definition}{Definitions}

\newcommand{\ECC}{\mathsf{ECC}}
\newcommand{\DEC}{\mathsf{DEC}}

\newcommand{\poly}{\text{poly}}
\newcommand{\polylog}{\text{polylog}}

\newcommand{\eps}{\varepsilon}

\newcommand{\enc}{\mathsf{enc}}

\newcommand{\LDC}{\mathsf{LDC}}

\newcommand{\conf}{\mathsf{conf}}

\newcommand{\estA}{\mathsf{estA}}
\newcommand{\alt}{\mathsf{alt}}

\newcommand{\bbN}{\mathbb{N}}
\newcommand{\bbR}{\mathbb{R}}

\newcommand{\bbF}{\mathbb{F}}
\newcommand{\bbE}{\mathbb{E}}

\newcommand{\cD}{\mathcal{D}}

\makeatletter
\newcommand{\customlabel}[2]{%
   \protected@write \@auxout {}{\string \newlabel {#1}{{#2}{\thepage}{#2}{#1}{}} }%
   \hypertarget{#1}{#2}
}
\makeatother

\newcommand\numberthis{\addtocounter{equation}{1}\tag{\theequation}}

\newcounter{casenum}

\newcounter{subcasenum}

\newcounter{casenump}

\newcommand{\casep}[2]{
    \ifthenelse{\equal{\value{casenump}}{0}}{
    \vskip.5\baselineskip\par\noindent
    }{}
    {\it Case \arabic{casenump}:} {\it #1}
    \vskip0.1\baselineskip
    \begin{addmargin}[1.5em]{1em}
    #2
    \end{addmargin}
    \addtocounter{casenump}{1}
}

\newcounter{subcasenump}

\begin{document}

\title{Error Correction for Message Streams}
\author{Meghal Gupta \thanks{E-mail:\texttt{meghal@berkeley.edu}. This author was supported by an UC Berkeley Chancellor’s fellowship.}\\UC Berkeley \and Rachel Yun Zhang\thanks{E-mail:\texttt{rachelyz@mit.edu}. This research was supported in part by DARPA under Agreement No. HR00112020023, an NSF grant CNS-2154149, and NSF Graduate Research Fellowship 2141064.}\\MIT}
\date{\today}

\sloppy
\maketitle
\begin{abstract}


In the setting of error correcting codes, Alice wants to send a message $x \in \{0,1\}^n$ to Bob via an encoding $\enc(x)$ that is resilient to error. In this work, we investigate the scenario where Bob is a \emph{low space} decoder. More precisely, he receives Alice's encoding $\enc(x)$ bit-by-bit and desires to compute some function $f(x)$ in low space. A generic error-correcting code does not accomplish this because decoding is a very global process and requires at least linear space. Locally decodable codes partially solve this problem as they allow Bob to learn a given bit of $x$ in low space, but not compute a generic function $f$. 

Our main result is an encoding and decoding procedure where Bob is still able to compute any such function $f$ in low space when a constant fraction of the stream is corrupted. More precisely, we describe an encoding function $\enc(x)$ of length $\poly(n)$ so that for any decoder (streaming algorithm) $A$ that on input $x$ computes $f(x)$ in space $s$, there is an explicit decoder $B$ that computes $f(x)$ in space $s \cdot \polylog(n)$ as long as there were not more than $\frac14 - \eps$ fraction of (adversarial) errors in the input stream $\enc(x)$.
\end{abstract}
\thispagestyle{empty}
\newpage
\tableofcontents
\pagenumbering{roman}
\newpage
\pagenumbering{arabic}

\section{Introduction}

Consider the following scenario: Alice streams a message denoted $x=x_1\ldots x_n$ to Bob that he receives and processes bit-by-bit. His goal is to compute some function $f(x_1\ldots x_n)$ unknown to Alice in significantly less space than necessary to store the entire stream. This scenario arises for instance in automatic control applications, where a device receives an incoming stream of data and needs to use it to make decisions.

As an example, a class of functions that Bob may wish to compute in small space from a message stream is the class of \emph{linear functions}. If Bob's function is $f(x) = f_y(x) = x \cdot y\mod{2}$ for some $y \in \{ 0, 1 \}^n$, then after receiving each bit of $x$, Bob adds $x_i \cdot y_i$ to a running sum. Bob only needs to track $i$ and the running sum modulo 2, which is in total $\log n + 1$ bits of space. 

However, this and other small space algorithms are very rigid when it comes to errors in the stream. Corruption of even one bit of Alice's message can change the output of Bob's linear function. The same applies if Bob's function is an index function, a majority, or the result of a decision tree.

In this work, we study what happens when there is noise in the incoming stream. In particular, we ask if it's possible to convert a given algorithm that processes a message stream into one that is robust to errors in the message while still allowing the output to be computed in low space.

In the usual error correction setting, Alice could just encode her message using a generic large distance error-correcting code. Bob would be able to compute any function $f$ in the presence of a small constant fraction of error, but he must receive and store the whole stream in order to decode, which is far too much space. Even if Alice sends a locally decodable code, which has the property that to determine a single bit of $x$ one only needs to track a few random bits of the codeword, it is not clear how Bob can compute a function requiring all $n$ bits without $\Omega(n)$ storage. Our question is whether there is an encoding that preserves the space complexity of Bob's original decoding algorithm while being resilient to error. 

In this work, we answer the question in the affirmative. Our main result is a scheme that protects any streaming algorithm against noise in the incoming stream. More precisely, we give an encoding $\enc(x)$ of length $O(n^{4+\delta})$ such that any streaming algorithm running in $s$ space on the noiseless stream can be converted to one running in $s \cdot \polylog(n)$ space on the encoded stream. It is correct with high probability whenever at most $\frac14 - \eps$ of the stream was corrupted. In the specific case where Bob's function is a linear function such as dot product, our encoding can be made to have length $O(n^{2+\delta})$.




\subsection{The Model}

We provide a formal definition of a noise resilient transformation for a message stream and processing algorithm. The transformation has two components:

\begin{itemize}
    \item An encoding function $\enc(\cdot) : X \subseteq \{ 0, 1 \}^n \rightarrow \{ 0, 1 \}^m$.
    \item An explicit transformation $B$ that takes as input a deterministic streaming algorithm $A : X \rightarrow \bbF_q$ and outputs a randomized streaming algorithm $B_A = B(A) : \{ 0, 1 \}^m \rightarrow \bbF_q$.
\end{itemize}
The algorithm is said to be $\alpha$-error resilient if whenever $\Delta(z, \enc(x)) < \alpha \cdot m$, then $B_A(z)$ outputs $A(x)$ with high probability.

\subsection{Our Results}

Our main result is a noise-resilient conversion for deterministic streaming algorithms.


\begin{restatable}{theorem}{mainthm}
\label{thm:main}
    Fix $\eps,\delta>0$. For any $X\subseteq \{0,1\}^n$, there is a noise resilient transformation consisting of an encoding $\enc : X\subseteq\{0,1\}^n \to \{0,1\}^m$, and a transformation of algorithms $B$ that is $\left( \frac14 - \eps \right) \cdot m$ error resilient with probability $\ge 1 - 2^{O_\eps(\log^2 n)}$. Moreover, $m = O_\eps(n^{4 + \delta})$, and for any deterministic algorithm $A$ on domain $X$ that runs in space $s$ and time $t$, the algorithm $B(A)$ runs in space $s \cdot O_\eps \left( (\log n)^{O(1/\delta)} \right)$ and time $m \cdot O_{\eps,\delta}\left( 1+ \frac{t}{n^2} \right)$.
    
\end{restatable}

In other words, given an algorithm $A$ that accepts a stream of length $n$ and uses $s$ space, we demonstrate a noise resilient algorithm $B(A)$ computing $A$ that uses $s \cdot \polylog(n)$ space. The transformation $\enc$ of the stream $x$ from the sender is independent of $A$ and has length $O(n^{4+\delta})$.

A priori, it is not obvious that there is \emph{any} low-memory noise resilient transformation of the algorithm $A$ even with an arbitrary blow-up in communication. For example, if the sender were to encode their message using an arbitrary error-correcting code, the receiver would need to store the entire message in memory in order to decode it before applying the algorithm $A$. Our result shows that not only does this low-memory transformation exist, but that it can be done efficiently (both with respect to communication complexity and computational complexity).

\paragraph{Linear Streaming Algorithms.} In the above transformation, the length of the resulting stream is $O(n^{4 + \delta})$ for any $\delta > 0$. In the specific case where the streaming algorithm is \emph{linear}, we construct a scheme where the stream length is only quadratic, namely $O(n^{2+\delta})$. The key property of linear streaming algorithms that we will be leveraging is that they can be computed in pieces in any order, and later computations do not depend on the results of previous ones.

\begin{restatable}{definition}{deflinear} [Linear Streaming Algorithms]
\label{def:linear-stream}
    A linear streaming algorithm $A : \{0,1\}^n \to \bbF_q$ is described by a list of functions $g_i : \{0,1\} \to \bbF_q$ for $i\in [n]$, and computes the value $A(x) = g_1(x_1)+\ldots + g_n(x_n)$ (where addition is over $\bbF_q$) by tracking the partial sum $g_1(x_1)+\ldots + g_i(x_i)$ upon receiving the $i$'th bit.
\end{restatable}

Note, for example, that every linear function on codomain $\{ 0, 1 \}^n$ admits a linear streaming algorithm. We note that linear streaming algorithms capture a large class of interesting functions, including linear sketching algorithms (see for example~\cite{CormodeM05,CharikarCF04,AhnGM12,JohnsonL84}).  For linear streaming algorithms, we show the following result.

\begin{restatable}{theorem}{thmlinear}\label{thm:main-linear}
    Fix $\eps,\delta>0$. There is a function $\enc : X\subseteq\{0,1\}^n \to \{0,1\}^m$ where $m = O_\eps(n^{2 + \delta})$ and an explicit transformation $B$ such that the following holds: For any linear streaming algorithm $A$ that takes as input $x\in X\subseteq \{0,1\}^n$ as a stream, runs in space $s$ and time $t$, and outputs $A(x)$, the algorithm $B_A = B(A)$ takes as input $z \in \{ 0, 1 \}^m$ as a stream, runs in space $s \cdot O_\eps((\log n)^{O(1/\delta)})$ and time $m \cdot O_{\eps,\delta}\left( 1+ \frac{t}{n^2} \right)$, and satisfies that whenever $\Delta(z, \enc(x)) < \left( \frac14 - \eps \right) \cdot m$ then $B_A(z)$ outputs $A(x)$ with probability $\ge 1 - 2^{O_\eps(\log^2 n)}$. 
\end{restatable}




\paragraph{Randomized Algorithms.} We remark that our transformations in Theorems~\ref{thm:main} and~\ref{thm:main-linear} 
are only for deterministic algorithms. However, this easily implies the result for algorithms $A$ that \emph{pre-sample} all of their randomness (that is, algorithms that fix their randomness before receiving any bits of $x$, at which point they are deterministic for the remainder of the algorithm). Such algorithms can be viewed as sampling from a distribution over deterministic algorithms. Our transformation can then be applied to the ensuing deterministic algorithm, thus correctly computing the function $A$ with high probability while being resilient to noise.

We remark that this notion of randomized algorithms whose only access to randomness is pre-sampled at the start of computation is quite general: randomized algorithms that sample its randomness online can often be converted to protocols where the randomness is pre-sampled, see e.g.~\cite{Nisan90}.

\paragraph{Non-Binary Input Alphabets.}

Our main theorems are stated for the setting where the input stream is binary. We remark that this assumption is without loss of generality: a larger alphabet stream can be converted to a binary one by assigning each alphabet symbol to a unique binary string.

\subsection{Discussion of Subsequent Work}

Since this paper was first posted, the work of~\cite{GuptaGS24} improved upon our schemes in a few major ways.
\begin{itemize}
    \item For the case of general deterministic streaming algorithms, they demonstrate a transformation that only requires near-quadratic encoding length.
    \item They show that this is essentially tight: there is no transformation that has sub-quadratic encoding length.
    \item And, in the case of linear streaming algorithms, they construct a transformation that has an encoding of near linear length.
\end{itemize}
This addresses most of the open questions posed by our paper. \mnote{the index problem isn't that important anyway, we can just restrict this section to disc of subsequent work and leave it at this.}

\subsection{Related Works}

\paragraph{Streaming Algorithms.}
The study of streaming algorithms began with the work Morris~\cite{Morris78} on approximate counting; a rigorous analysis was given later in~\cite{Flajolet85}. Later works introduced other classic problems in streaming, including heavy hitters~\cite{CharikarCF02}, $\ell_p$ approximation~\cite{AlonMS96, MonemizadehW10, IndykW05}, and finding a nonzero entry in a vector (for turnstile algorithms)~\cite{MonemizadehW10}.

Many works, for example \cite{Garg21, Chen16, Ma21, BenJWY22}, consider the problem of processing noisy data sets using streaming algorithms. \cite{Garg21} shows a memory lower bound for learning a parity function with noisy samples of random linear equations.

However, this type of noise is quite different from our treatment. In these problems, noise is inherent in the way samples are generated, whereas we are investigating noise that occurs in the communication process. 

\paragraph{Error Correcting Codes.}

Our result can be viewed through the lens of low-space noise resilient one-way communication.

Noise resilient one-way communication, in other words, error correction~\cite{Shannon48,Hamming50}, is one of the most well-studied problems in computer science. One specific line of works related to our result is that of \emph{locally decodable codes}. Locally decodable codes \cite{Yekhanin12} can be viewed as a low-time and low-space version of error-correcting codes, where the goal is to learn a single bit of the original message. By constrast, for us, the decoder must be able to piece together any function of the input that is noiselessly computable in low space, with the tradeoff that the decoder accesses the entire encoding via a stream rather than via query access. Local decodable codes have been constructed in a variety of different parameter regimes, including constant query~\cite{Efremenko09, DvirGY11} and near $1$ rate~\cite{KoppartySY14}. In our work, we will use Reed-Muller codes~\cite{Muller54, Reed54} that achieve $\polylog(n)$ query complexity and slightly super-linear block length.

Another line of works related to ours is that of streaming algorithms for local testing of codes~\cite{RudraU10, McGregorRU11}. However, this direction was rendered moot by the recent discovery of locally testable codes with constant rate, distance, and locality~\cite{DinurELM22}.




\paragraph{Coding for Interactive Communication.}
The analogue of noise resilient one-way communication in the interactive setting, where Alice and Bob have inputs $x,y$ and engage in a protocol to compute some $f(x,y)$ while being resilient to error, comprises the study of \emph{coding for interactive communication}. Interactive coding was studied starting with the seminal works of Schulman~\cite{Schulman92,Schulman93,Schulman96} and continuing in a prolific sequence of followup works, including~\cite{BravermanR11,Braverman12,BrakerskiK12,BrakerskiN13,Haeupler14,BravermanE14,DaniHMSY15,GellesHKRW16,GellesH17,EfremenkoGH16,GhaffariH13,GellesI18,EfremenkoKS20b,GuptaZ22c,GuptaZ22a}. We refer the reader to an excellent survey by Gelles~\cite{Gelles-survey} for an extensive list of related work. 

The recent work of~\cite{EfremenkoHKKRS23} studies the space complexity of interactive coding schemes. Their main result is an interactive coding scheme that preserves the space complexity of the original noiseless protocol up to a $\log(n)$ factor, where $n$ is the length of the protocol. Their result can be viewed as the interactive analogue of ours, where their noiseless and noise resilient protocols are both interactive. We remark that their techniques are quite different than ours and do not apply to our setting, however, since their approach crucially relies on the communication of feedback from Bob to Alice.

\section{Overview of Techniques} \label{sec:overview}

Consider the task of computing a linear function $f(x) = \sum_{i \in [n]} g_i(x_i)$. In a noiseless setting, one can compute $f(x)$ by tracking the partial sum $\sum_{1 \le i \le n'} g_i(x_i)$ and updating it with $g_{n'+1}(x_{n'+1})$ upon receiving the $(n'+1)$'th bit of $x$.

As a first attempt to create a noise resilient version of this algorithm, one might consider the input stream $\enc(x) = \LDC(x)$, where $\LDC$ is a locally decodable code. In short, a locally decodable code is a noise resilient encoding of $x$ such that in order to decode any specific bit of $x$, say $x_i$, the decoder simply has to read $\polylog(n)$ (randomized) bits of $\LDC(x)$. In a streaming setting, the decoder can record only these $\polylog(n)$ bits into memory, then run the decoding algorithm to determine $x_i$ after all such query bits have been seen.

However, it's not clear that this local decoding property extends to arbitrary linear functions.\footnote{The smallest local decodable code we know of that supports local decoding to any arbitrary linear function is an extension of the Hadamard code, which has block length exponential in $n$. This will not be a satisfactory solution for us because (a) we would like encodings that are efficiently computable, and (b) keeping track of which bit of the stream one is receiving takes $n$ memory.} If we attempt to simultaneously decode all bits $x_i$ from $\LDC(x)$ using a separate query set $Q_i$ for each index, the issue is that we will need to track $\Omega(n)$ queries simultaneously since we cannot ensure that the local decoding for any index finishes before another: the query sets $Q_i$ are typically randomized so that any bit of the codeword is equally likely to belong to the query set for any particular index. 

A second attempt is to send $\LDC(x)$ $n$ times, with the intent that Bob locally decodes $x_i$ in the $i$'th $\LDC(x)$ and computes $f(x)$ by tracking the partial sum $\sum_{1 \le i \le n'} g_i(x_i)$ at the end of the $i$'th chunk. Now, the only space requirements are the space required to store this partial sum and the space required to locally decode $x_i$ in each $\LDC(x)$, which is $\polylog(n)$. However, the adversary can simply corrupt the first $\LDC(x)$ so that Bob mis-decodes $x_1$, thereby corrupting his final output, and ultimately, this approach is not so different from just sending the bits $x_1\ldots x_n$ in order. 

What Bob would like to do is randomize which $x_i$ he is computing from a given copy of $\LDC(x)$ so that the adversary does not know. This crucially uses the property that LDC's allow Bob to decode \emph{any} $x_i$ not known to the encoder or adversary. Then, if he computes each $x_i$ many times and takes the majority, he will learn $x_i$ correctly with high probability, regardless of the adversary's corruption process. Since the adversary does not know in advance which chunks Bob is computing $x_i$ in, she cannot target corruptions towards a specific $x_i$. However, this runs into an issue with the amount of space Bob requires. Since he computes each $x_i$ many times throughout the protocol and only takes the majority when he has finished these computations, he needs to track his current guess for its value through (almost) the entire protocol. When each $x_i$ was computed only once, Bob could just add $g_i(x_i)$ to the running sum, but now he needs to track each $x_i$ individually between the first and last time he queries it, for a total of $\Omega(n)$ bits of space.

The crux of our construction is an organized system for aggregating decoding attempts in space much smaller than storing all $n$ attempts simultaneously.

\paragraph{A recursive computation.} 

Let $r = \polylog(n)$. Suppose that we already know how to compute linear functions on inputs of size $n/r$ in space $s_{n/r}$, that is resilient to $\eps_{n/r}$ fraction of errors in the stream of $M_{n/r}$ copies of $\LDC(x)$. We will use this to build a scheme for computing linear functions on inputs of size $n$. Note that the base case is given by computing linear functions on a single bit, which can be done by local decoding on a single $\LDC(x)$.

To compute $f(x) = \sum_{i \in [n]} g_i(x_i)$, we just need to compute $f_1(x) = \sum_{i \in [n/r]} g_i(x_i), ~f_2(x) = \sum_{i \in [n/r+1, 2n/r]} g_i(x_i),~ \dots, ~f_r(x) = \sum_{i \in [n-n/r+1, n]} g_i(x_i)$. Each sub-computation can individually be done with the guarantee that if $< \eps_{n/r}$ fraction of the stream is corrupted, then the function is computed correctly.

Consider a stream consisting of $\ell \cdot r \cdot  M_{n/r}$ copies of $\LDC(x)$ (you can think of $\ell = \log^2 n$). We split up this stream into $\ell$ chunks each consisting of $r$ blocks of $M_{n/r}$ $\LDC(x)$'s. In each of the $\ell$ chunks, we will assign the $r$ blocks of $M_{n/r}$ $\LDC$'s to the computation of a random permutation of $f_1(x), \dots, f_r(x)$. Throughout the algorithm, we will keep track of all the outputs from each of the $\ell$ sub-computations for each of $f_1(x), \dots, f_r(x)$. At the end, we can take a majority vote for each of the $f_j(x)$. 
We illustrate the recursion process in Figure~\ref{fig:tree-figure}.

\begin{figure}[h!]
\centering

\begin{tikzpicture}[level/.style={level distance=25mm}, decoration={brace,mirror,amplitude=7}]

    \begin{scope}[xshift=-4cm]
        \node {$f$}
            child {node [xshift = -5mm] (1) {$f_1$}}
            child {node [xshift = -5mm] (2) {$f_2$}}
            child {node [xshift = 5mm] (r) {$f_r$}};

            \path (2) -- node[auto=false]{\textbf\ldots} (r);
    \end{scope}
    
    \begin{scope}[xshift=4cm]
        \node {$f$}
            child {node [xshift=-2mm] (1-1) {$f_{\pi_1(1)}$}}
            child {node [xshift=-4mm] (1-2) {$f_{\pi_1(2)}$}}
            child {node [xshift=-4mm] (1-r) {$f_{\pi_1(r)}$}}
            child {node [xshift=4mm] (2-1) {$f_{\pi_\ell(1)}$}}
            child {node [xshift=0mm] (2-2) {$f_{\pi_\ell(2)}$}}
            child {node [xshift=2mm] (2-r) {$f_{\pi_\ell(r)}$}};

            \path (1-r) -- node[auto=false, yshift=5mm]{\textbf\ldots} (2-1);
            \path (1-2) -- node[auto=false]{\ldots} (1-r);
            \path (2-2) -- node[auto=false]{\ldots} (2-r);

            \draw [thick] ([yshift=-4mm]1-1.west) to [square right brace] ([yshift=-4mm]1-r.east);
            \draw [thick] ([yshift=-4mm]2-1.west) to [square right brace] ([yshift=-4mm]2-r.east);
            \draw [decorate, thick] ([yshift=-8mm]1-1.west) --node[below=3mm]{$\ell$} ([yshift=-8mm]2-r.east);
    \end{scope}

    \draw[Bar->,line width=1pt] (-1.8,-1) -- (0,-1);
    
\end{tikzpicture}



\caption{In our transformation, each sub-function is computed $\ell$ times instead of $1$, and in each of the $\ell$ chunks, the sub-functions $f_i$ are computed in according to a random permutation $\pi$.}
\label{fig:tree-figure}
\end{figure}
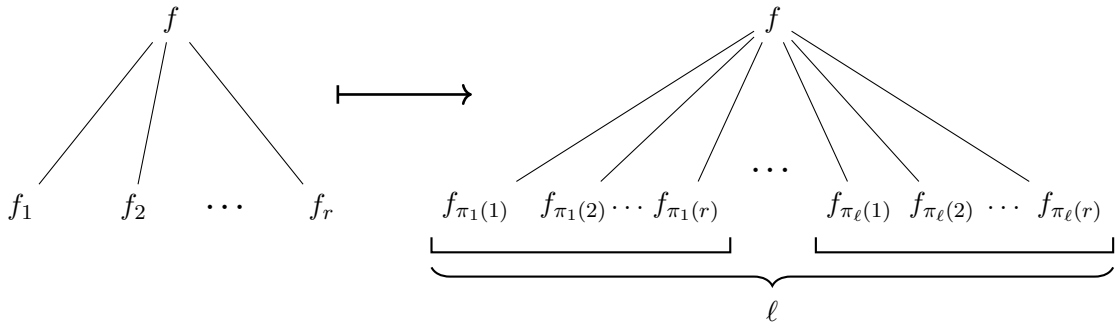

Because of the way we've randomly assigned blocks to $f_1, \dots, f_r$, the adversary cannot concentrate her attack on the computation of any specific $f_i$. More precisely, we can show that the amount of error present in the blocks corresponding to the computation of $f_k$ is (roughly) proportional to the total error present. Then, as long as at least half the blocks corresponding to $f_k$ have $< \eps_{n/r}$ fraction of errors, we have that the majority vote for each of $f_1, \dots, f_r$ will be correct. In total, this means that our algorithm for computing $f(x)$ is resilient to $\frac{\eps_{n/r}}2$ fraction of errors. 

Unfortunately, this factor of $2$ loss in the error resilience in each iteration of the recursion will prove to be problematic, as we can no longer guarantee a positive error resilience when $n$ gets large.

\paragraph{Decoding with confidences.}

In the above analysis, our worst case error distribution was that where just under one-half of the chunks have $\eps_{n/r}$ error while the rest are error-free. In this case, the chunks with $\eps_{n/r}$ error will produce a wrong output, while the error-free ones will produce a correct output.

However, we notice that in certain local decoding constructions (for us, we will use a Reed-Muller code), one can actually obtain an estimate of how much error there was in the codeword based on the inconsistency of the queries. In particular, in chunks where there were $\eps_{n/r}$ errors, even though we will obtain a wrong answer, we can also see that there were a lot of errors in the codeword, and thereby be less sure of the obtained answer.

This motivates our definition of \emph{local decoding with confidences}: in addition to outputting an answer, the decoding algorithm should also output a \emph{confidence} $c \in [0, 1]$ indicating how confident it is that the answer is correct. A confidence of $0$ indicates that it's not sure at all: the fraction of errors seen was near the threshold $\eps_1$; whereas a confidence of $1$ indicates that it's quite sure: either the answer is right and no errors occurred, or at least $2\eps_1$ errors were required to obtain a wrong answer with confidence $1$. Another way to view the confidence is as an estimate of how much error was present in the total codeword (how far the actual error is from $\eps_1)$.

Now, when we perform the recursive step to obtain a noise resilient algorithm for $f$ from $f_1, \dots, f_r$, in addition to recording the outputted answer from each chunk, we will also record the associated confidence. Suppose that the $\ell$ answers and confidences obtained for $f_k$ were $(\hat{q}^{(1)}, \hat{c}^{(1)}), \dots, (\hat{q}^{(\ell)}, \hat{c}^{(\ell)})$. Then, the weighted majority vote of these answers is defined to be the $\hat{q}$ with the largest total confidence. We can show that as long as the cumulative fractional error $c_n$ is at most $c_{n/r}$, we're guaranteed that the correct answer has more the highest total confidence.

We remark that one does not have to store all $\ell$ pairs $(\hat{q}^{(j)}, \hat{c}^{(j)})$. Rather, we can keep track of a most likely $q$ and associated confidence $c$. To update $(q, c)$ with a new pair $(\hat{q}, \hat{c})$, one can update $c$ by $\pm \hat{c}$ ($+\hat{c}$ if $\hat{q} = q$, and $-\hat{c}$ otherwise). If the resulting confidence $c$ is less than $0$, then it means that the answer $\hat{q}$ is supported by less overall error, so we set our most likely $q$ to be $\hat{q}$ and flip the sign of $c$.

Now, we analyze the space and communication complexity of the resulting algorithm. The space required is simply $s_{n/r}$, the space required to do a size $n/r$ sub-computation, plus the space required to store a pair $(q, c)$ for each of $f_1, \dots, f_r$. Overall, this means that for each layer of the recursion, we are only gaining an additive factor of $r \cdot (|q| + |c|)$ space overhead, so the total space overhead is at most polylogarithmic in $n$. As for the size of the resulting stream, we see that the number of $\LDC(x)$'s is $\approx (r \cdot \ell )^{\log_r n}$,\footnote{In our actual protocol, we will repeat the procedure $\log^2 n$ times and take a final majority vote, so the number of $\LDC(x)$'s used will be $(r \cdot \ell)^{\log_r n} \cdot \log^2 n$.} which is polynomial in $n$. 

\paragraph{Computing non-linear functions.}


Our algorithm so far has only captured functions $f$ that are splittable into $r$ sub-computations whose computations do not depend on each other. However, for a general function, this may not be the case. Indeed, for many functions, the computations must be done \emph{sequentially}, where later computations can only be done once earlier computations have finished. We may split a general function $f$ into $r$ sub-functions where $q_1:=f_1(x[1:n/r],\emptyset)$ is the state of the algorithm after $n/r$ bits of the stream, $q_2:=f_2(x[n/r+1:2n/r],q_1)$ is the state of the algorithm after receiving the next $n/r$ bits of the stream, and so on. The difference from the linear function case is that to compute $f_j$, we must have the correct output of $f_{j-1}$.

In order to handle such sequential computations, we modify our above algorithm as follows. In the recursive layer computing $f$ from functions $f_1, \ldots, f_r$, we always compute $f_j$ from the starting state $q_{j-1}$ which is our current best guess for the output of $f_{j-1}$. This value of $q_{j-1}$ may not always be correct, in which case the computation done for $f_j$ will be doomed. 

As with our algorithm for linear functions, after every chunk, we update the most likely state $q_j$ and confidence $c_j$ for the sub-function that was computed in that chunk. However, when a state $q_j$ changes, this also means that computations done for $f_{j+1}\ldots f_n$ were based on wrong information. In the case of linear functions, the computations are independent so this does not matter, but for sequential functions, we have to discard the guesses $q_{j+1}\ldots q_n$. Formally, we set these states to $\emptyset$ and set the corresponding confidences all to $0$. 

Because each sub-function $q_j$ is only computed usefully when $q_{j-1}$ is correct and also the value of $q_j$ is frequently discarded, it is less clear that over many iterations the $q_j$'s will converge to the correct values. Nevertheless, we prove in Section~\ref{sec:general-streaming} that over enough iterations, the guesses $q_j$ will eventually all be correct as long as there are not too many errors in the stream.


\paragraph{Outline.}

The rest of this paper is organized as follows. In Section~\ref{sec:prelims}, we start with some preliminary theorems about error-correcting codes. In Section~\ref{sec:ldc}, we define our notion of local decoding with confidences and prove that Reed-Muller codes satisfy this property. Then, in Sections~\ref{sec:linear} and~\ref{sec:general-streaming}, we present our noise resilient streaming algorithms for linear and general (sequential) functions, respectively.





\section{Preliminaries} \label{sec:prelims}

\paragraph{Notation.}
\begin{itemize}
    \item The function $\log$ is in base $2$ unless otherwise specified.
    \item The set $[n]$ denotes the integers $1\ldots n$.
    \item We use $x[i:j]$ to denote the $i$'th to $j$'th bits of $x$ inclusive. We use $x(i:j]$ to refer to the $(i+1)$'th to $j$'th bit inclusive.
    \item We use $e$ as a variable rather than the universal constant $e$.
\end{itemize}

\subsection{Error-Correcting Codes} \label{sec:ecc}

We begin with some results about error-correcting codes. The first is the existence of a relative distance $\frac12$ binary code.

\begin{theorem} \cite{GuruswamiS00} \label{thm:binary-ecc}
    For all $\epsilon > 0$, there exists an explicit error-correcting code $\ECC_\epsilon = \{ \ECC_{\epsilon, n} : \{ 0, 1 \}^n \rightarrow \{ 0, 1 \}^m \}_{n \in \bbN}$ with relative distance $\frac12 - \epsilon$ and $m = O_\epsilon (n)$, and a $\poly_\epsilon(n)$-time decoding algorithm $\DEC_{\epsilon} : \{ 0, 1 \}^m \rightarrow \{ 0, 1 \}^n$, such that for any $x \in \{ 0, 1 \}^n$ and $w \in \{ 0, 1 \}^m$ satisfying $\Delta(\ECC_\eps(x), w) < \left( \frac14 - \frac\eps4 \right) \cdot m$, it holds that $x = \DEC_\eps(w)$.
\end{theorem} 

Our second theorem is about the efficient decoding of Reed-Solomon codes~\cite{ReedS60}.

\begin{theorem}[Berlekamp-Welch; see e.g.~\cite{GemmellS92}] \label{thm:berlekamp-welch}
    Given a collection of tuples $(\alpha_1, v_1), \dots, (\alpha_n, v_n)$ where $\alpha_i, v_i \in \bbF_q$ and a parameter $d \le n$, there is an algorithm running in $\poly(n, \log q)$ time that outputs a polynomial $g(\alpha) \in \bbF_q[\alpha]$ of degree at most $d-1$ such that $\Delta( ( g(\alpha_i) )_{i \in [n]}, ( v_i )_{i \in [n]}) < \frac{n-d+1}{2}$.
\end{theorem}

Finally, we recall generalized minimum distance decoding, which allows us to decode concatenated codes efficiently.

\begin{theorem}[GMD Decoding \cite{Forney66,Guruswami06}] \label{thm:concat-RS-decoding} 
    Given two codes, $N_{outer} : \{ 0, 1 \}^k \rightarrow \Sigma^n$ of distance $1 - \eps_{outer}$ with decoding time $T_{outer}$, and $N_{inner} : \Sigma \rightarrow \{ 0, 1 \}^m$ of distance $\frac12 - \eps_{inner}$ with decoding time $D_{inner}$, there is an algorithm running in time $\poly(D_{inner}, D_{inner})$ that on input $w \in \{ 0, 1 \}^{n \cdot m}$ outputs $x \in \{ 0, 1 \}^k$ such that $\Delta(w, N_{outer} \circ N_{inner} (x)) < \frac{(1-\eps_{outer})(\frac12 - \eps_{inner})}{2}$.

\end{theorem}

\section{Locally Decoding with Confidences} \label{sec:ldc}

We now recall the definition of \emph{locally decodable codes}. For us, besides just correctness of local decoding when the distance to a codeword is small, we require that the decoder output a \emph{confidence} indicating how much error it witnessed (less error means higher confidence).

\begin{theorem} (Locally Decodable Code) \label{thm:ldc}
    For any $\eps > 0$ and $d = d(n) = \log^{1/\delta} n$, there is a code $\LDC : \{ 0, 1 \}^n \rightarrow \{ 0, 1 \}^{N(n)}$ where $N(n) = O_\eps (n^{1+\delta})$ 
    that satisfies the following properties:
    \begin{itemize}
    \item {\bf Distance:} 
        For any $x \not= y \in \{ 0, 1 \}^n$, it holds that $\Delta(\LDC(x), \LDC(y)) \ge \left( \frac12 - \eps \right) \cdot N$.
    \item {\bf Locally Decoding with Confidence:} 
        For any index $i \in [n]$ and any word $w \in \{ 0, 1 \}^N$, there exists a randomized algorithm $\cD_i$ that reads $Q = O_\eps ( \log^{2/\delta} n )$
        bits of $w$, runs in $O_\eps (n^{O(1/\delta)})$ time, and outputs a bit $\hat{x}_i$ along with confidence $\conf \in [0, 1]$, satisfying that 
        \begin{itemize}
        \item 
            For any $x \in \{ 0, 1 \}^n$ such that $\Delta(w, \LDC(x)) = \left( \frac14 - 2\eps - e \right) \cdot N$ where $e \ge 0$, it holds that 
            \[
                \hat{x}_i = x_i \qquad \text{and} \qquad \conf > e
            \]
            with probability at least $1 - \exp\left(- \log^2 n \right)$.
        \item 
            For any $x \in \{ 0, 1 \}^n$ such that $\Delta(w, \LDC(x)) = \left( \frac14 - 2\eps + e \right) \cdot N$ where $e > 0$, it holds that 
            \[
                \hat{x}_i = x_i \qquad \text{or} \qquad \conf < e
            \]
            with probability at least $1 - \exp\left( - \log^2 n \right)$.
        \end{itemize}
        Furthermore, the confidence $\conf$ can be represented as a rational number with denominator $4Q$. 
    \end{itemize}
\end{theorem}


\begin{proof}
    The locally decodable code that we will be using is the concatenation of two codes. The outer code $C_{outer}$ is a Reed-Muller code of $m = \log_d n$-variate polynomials over $\bbF = \bbF_q$ of degree $d-1$ where $q \in [2d / \eps, 4d / \eps]$ is a power of $2$. $C_{outer}$ has dimension $d^m = n$ and block length $N_{outer} = q^m < (4d / \eps)^m = n \cdot \left( \frac4\eps \right)^m = n^{1 + \delta \cdot \log(4/\eps)/ \log\log n} = O_\eps(n^{1 + \delta/2})$.  The inner code $C_{inner}$ is a binary code of relative distance $\frac12 - \frac\eps2$ (see Theorem~\ref{thm:binary-ecc}) with block length $N_{inner} = O_\eps(\log q) = o_\eps(n^{\delta/2})$. Thus, the block length of the concatenated code $C$ is $N = N_{outer} \cdot N_{inner} = O_\eps(n^{1+\delta})$. We will index the bits of the concatenated code by pairs $(v, j)$ where $v \in \bbF^m$ and $j \in [N_{inner}]$. We remark that the outer code $C_{outer}$ is systematic. 
    
    By the Schwartz-Zippel lemma~\cite{Schwartz80,Zippel79}, the distance of the outer code is $\ge (q - d + 1) \cdot q^{m-1} \ge \left( 1 - \frac\eps2 \right) \cdot N_{outer}$. So the relative distance of the concatenated code is $\ge \left( 1 - \frac\eps2 \right) \cdot \left( \frac12 - \frac\eps2 \right) \ge \frac12 - \eps$.

    The decoding algorithm is as follows. Suppose we want the $i$'th index of $x$ which is located at $v_0 \in \bbF^m$ in the outer code. Then, we do the following $k = 10\log^2 n / \epsilon^2$ times: Pick a random degree-$2$ polynomial $p(\lambda) = v_0 + v_1 \lambda + v_2 \lambda^2$ by sampling $v_1, v_2 \gets \bbF^m$. Then, we query all bits of $w$ located at $(p(\lambda), j)$ for $\lambda \in \bbF\backslash \{ 0 \}$, $j \in [N_{inner}]$ (we denote these collection of values by $w|_p$). This means that $Q = k \cdot (q-1) = O_\eps(\log^{2/\eps})$. By Theorem~\ref{thm:concat-RS-decoding} (combined with Theorems~\ref{thm:binary-ecc} and~\ref{thm:berlekamp-welch}), one can in $\poly(q \log q) = O_\eps(n^{O(1/\delta)})$ time find the unique degree $2d-2$ polynomial $h \in \bbF[\lambda]$ such that the distance between $w|_p$ and the $C_{inner}$ encodings of $h(\lambda)$ is at most $(q - 2d + 1) \cdot \left( \frac12 - \frac\eps2 \right) \cdot N_{inner} / 2 \ge \left( \frac12 - \eps \right) \cdot \frac{Q}{k}$. Then, the decoding algorithm outputs $\hat{x}_i$ set to be the majority of all the values of $h(0)$ among the $k$ iterations (where if such a polynomial $h$ of distance $(q-2d+1) \cdot \left( \frac12 - \frac\eps2 \right) \cdot N_{inner}/2$ doesn't exist, we say that $h(0)$ evaluates to $0$).

    The confidence $\conf$ is calculated as follows. For each of the $k$ polynomials $p$, we let $\Delta^h_p$ be the distance of the queried bits from the $C_{inner}$ encodings of $h(\lambda)$. If $h(0) = \hat{x}_i$, we set $\Delta_p = \Delta^h_p$, and otherwise we set $\Delta_p = (q-2d+1) \cdot \frac12 - \frac\eps2 ) \cdot N_{inner} / 2 - \Delta^h_p$. This gives us an estimate $\Delta_p$ of the amount of error present in $w$ conditioned on the output guess $\hat{x}_i$ being the actual value of $p(0)$. In particular, $\Delta_p \le \Delta(w|_p, \LDC(x)|_p)$. Then, we set $\conf := \frac14 - \sum_p \Delta_p / Q$.
    

    We now discuss why the local decoding properties hold. First, suppose $x \in \{ 0, 1 \}^n$ satisfies that $\Delta(w, \LDC(x)) \le \left( \frac14 - 2\eps \right) \cdot N$. Then, 
    \begin{align*}
        \Pr& \left[ \hat{x}_i \not= x_i ~\text{or}~ \conf \le \frac14 - 2\eps - \frac{\Delta(w, \LDC(x))}{N}  \right] \\
        &\le \Pr \left[ \hat{x}_i \not= x_i \right] + \Pr \left[ \sum_p \Delta_p \ge \left( \frac{\Delta(w, \LDC(x))}{N} + 2\eps \right) \cdot Q ~\Bigg|~ \hat{x}_i = x_i \right] \\
        &\le e^{-\eps^2k/10},
    \end{align*}
    where we've used the following two inequalities:
    \begin{itemize}
    \item 
        By Proposition 2.7 in~\cite{Yekhanin11}, the probability that for any of the $k$ polynomials $p$ it holds that $p(0) \not= x_i$ is 
        \[
            \le \frac{4(\Delta - \Delta^2)}{(q-1) \cdot \left( (1 - \eps) \cdot \left( \frac12 - \frac\eps2 \right) - 2\Delta \right)^2}
            \le \frac{1}{\eps (q-1)} 
            \le \frac{1}{\log^{1/\delta} n}
        \]
        when $\Delta := \frac{\Delta(c, \LDC(x))}{N} \le \frac14 - 2\eps$.
        Then, by the Chernoff bound, 
        \begin{align*}
            \Pr[\hat{x}_i \not= x_i] 
            \le \Pr \left[ \# \left[ p : p(0) \not= x_i \right] \ge \frac{k}2 \right] 
            \le e^{-k/10}.
        \end{align*}
    \item 
        Given that $\hat{x}_i = x_i$, we have that $\Delta_p \le \Delta(w|_p, \LDC(x)|_p)$, so that $\bbE [\Delta_p] \le \frac{\Delta(w, \LDC(x))}{N} \cdot \frac{Q}{k}$. By Hoeffding's inequality, 
        \[
            \Pr \left[ \sum_p \Delta_p \ge \left( \frac{\Delta(w, \LDC(x))}{N} + 2\eps \right) \cdot Q \right]
            < e^{-8\eps^2 \cdot k}.
        \]
    \end{itemize}

    \noindent
    Second, suppose $x \in \{ 0, 1 \}^n$ satisfies that $\Delta(w, \LDC(x)) > \left( \frac14 - 2\eps \right) \cdot N$. Then, 
    \begin{align*}
        \Pr & \left[ \hat{x}_i \not= x_i ~\text{and}~ \conf \le \frac{\Delta(w, \LDC(x))}{N} - \left( \frac14 - 2\eps \right) \right] \\
        &\le \Pr \left[ \conf \le \frac{\Delta(w, \LDC(x))}{N} - \left( \frac14 - 2\eps \right) ~\Bigg|~ \hat{x}_i \not= x_i \right] \\
        &= \Pr \left[ \sum_p \Delta_p \ge \left( \frac12 - 2\eps - \frac{\Delta(w, \LDC(x))}{N} \right) \cdot Q ~\Bigg|~ \hat{x}_i \not= x_i \right].
    \end{align*}
    Note that since $\hat{x}_i \not= x_i$, and the restriction of $\LDC(x)|_p$ to any polynomial $p$ is a code of distance $\ge (q - 2d + 1) \cdot \left( \frac12 - \frac\eps2 \right) \cdot N_{inner} \ge \left( \frac12 - \eps \right) \cdot \frac{Q}{k}$. Thus, if $\Gamma_p$ denotes the actual amount of corruptions to the points of $\LDC(x)|_p$, then $\Gamma_p + \Delta_p \ge \left( \frac12 - \eps \right) \cdot \frac{Q}{k}$. The above inequality is then
    \begin{align*}
        \le \Pr \left[ \sum_p \Gamma_p \ge \left( \frac{\Delta(w, \LDC(x))}{N} + \eps  \right) \cdot Q \right].
    \end{align*}
    Since $\bbE[\Gamma_p] = \frac{\Delta(w, \LDC(x))}{N} \cdot \frac{Q}{k}$, by Hoeffding's Inequality,
    \[
        \Pr \left[ \sum_p \Gamma_p \ge \left( \frac{\Delta(w, \LDC(x))}{N} + \eps \right) \cdot Q \right]
        \le e^{-2\eps^2 \cdot k }.
    \]
    
\end{proof}

\section{Noise Resilient Streaming for Linear Algorithms} \label{sec:linear}

Our first result is a noise resilient conversion for linear streaming algorithms $A$. We begin by recalling the definition of a linear streaming algorithm.

\deflinear*

For such algorithms, we describe a noise resilient conversion that incurs quadratic blow-up in communication complexity. 

\thmlinear*

\subsection{Statement of Algorithm}

Let $\eps > 0$. Throughout the section let $r = \log^{1/\delta} n$ and define $\ell := \log^2{n}$. We also assume for simplicity that $\log_r n$ is an integer. We will also assume that $n > \max \left\{ \exp(\exp({8\delta/\eps})), \exp(1/\eps) \right\}$ is sufficiently large, so that also $\ell < n$.

We begin by specifying the encoding $\enc(x)$ that Alice sends in the stream.

\begin{definition}[$\enc(x)$] \label{def:enc-linear}
        Let $\LDC(x)$ be a locally decodable code with $N(n) = |\LDC(x)| = O_\eps(n^{1 + \delta})$ and query complexity $Q = O_\eps(\log^{2/\delta} n)$ satisfying the guarantees of Theorem~\ref{thm:ldc} for $\eps/8$. Then Alice sends $\enc(x):=\LDC(x)^{\log^2 n \cdot M_n}$ where $M_n = (r \cdot \ell)^{\log_r n}$ (that is, $\enc(n)$ is $\log^2 n \cdot M_n$ copies of $\LDC(x)$). In particular, $m = |\enc(x)| = (r \cdot \ell)^{\log_r n} \cdot \log^2 n \cdot N(n) = O_\eps(n^{2+4\delta})$.
\end{definition}

Throughout this section, we use the notation $A_{i,j}(x)$ to denote the quantity $g(x_{i+1})+\ldots + g_j(x_j)$. Then, $A(x)=A_{0,n}(x)$. We define $N_{j-i} := (r \cdot \ell)^{\log_r(j-i)} \cdot N(n)$ which, as we'll see in the description below, represent the number of bits read by the algorithm to approximate $A_{i,j}(x)$.

We describe Bob's streaming algorithm $B_A:=B(A)$. Before stating it formally, we provide a high-level description.

\paragraph{Description:} 
Let $z$ be the incoming stream that Bob receives. At a high level, the algorithm $B_A$ runs by generating many guesses for $A(x)$ along with confidences indicating how likely that guess was correct (depending on how much error is witnessed while generating the guess). At the end, the guess that has the most cumulative confidence is outputted.

For $i,j$ where $j-i$ is an integer power of $r$, we define the algorithm $\estA(i, j)$ that takes in two input two indices $i, j \in [n]$. It reads $N_{j-i}$ many bits of $z$ after which it outputs a guess for $A_{i,j}(x)$ along with a confidence. In particular, $\estA(0, n)$ outputs a guess for $A(x)$ along with a confidence. By running $\estA(0, n)$ many times and aggregating the guesses weighted by the confidences, we obtain a final output.

To compute $\estA(i, j)$, we break down the interval $(i : j]$ into $r$ subintervals $(i_0 : i_1], (i_1 : i_2], \dots, (i_{r-1} : i_r]$ where $i_a = i + a \cdot \frac{j-i}{r}$. We then recursively formulate guesses for each $A \left( x\left(i_{a-1}:i_a\right] \right)$ over many iterations by calling $\estA(i_{a-1},i_{a)}$ each many times. The choice of $a$ in each iteration must be randomized, so that the adversary cannot attack any single choice of $a$. More specifically, we will split up the $N_{j-i}$ length input stream into $\ell$ chunks, each of length $N_{j-i} / \ell$ bits. Each chunk is split into $r$ sections, each of size $N_{j-i} / (\ell \cdot r) = N_{(j-i)/r}$. In each chunk, we pick a random permutation $(\pi_1, \pi_2, \dots, \pi_r)$ of $[r]$. In the $a$'th section of the chunk, we will compute on the subproblem corresponding to interval $\left( i_{\pi_a-1} : i_{\pi_a} \right]$ by calling $\estA \left( i_{\pi_a-1}, i_{\pi_a} \right)$. 


Whenever a recursive call to $\estA \left( i_{a-1}, i_a \right)$ is completed, outputting $(\hat{q}, \hat{c})$, then: if $q_a = \hat{q}$ then the confidence $c_a$ is increased by $\hat{c}$, and if $q_a \not= \hat{q}$ then the confidence $c_a$ is decreased by $\hat{c}$. If the confidence $c_a$ becomes negative, then we are more sure that the correct state is $\hat{q}$ rather than $q_a$: then $q_a$ is replaced by $\hat{q}$ and the confidence is negated (so that it's positive).

\begin{algorithm}[H]
\caption{Bob's Noise Resilient Algorithm $B_A$}
\label{alg:main-linear}
\renewcommand{\thealgorithm}{}
\floatname{algorithm}{}

\begin{algorithmic}[1]
\State \textbf{input} $n\in \bbN$ and stream $z \in \{ 0, 1 \}^m$.
\Function{$\estA$}{$i,j$} \Comment{compute the state of $A$ starting at state $q$ between steps $i$ and $j$}

\If{$j=i+1$}
\State Read the next $|\LDC(x)|$ bits of the stream $y$. \label{line:ldc-linear}
\State Using Theorem~\ref{thm:ldc} compute guess $\hat{b}$ for $x[j]$ and confidence $\hat{c}$ in $O_\eps\left((\log n)^{O(1/\delta)} \right)$ bits of space. 
\State \Return $g_i(\hat{b}), c$.

\Else
\State Let $i_0=i, i_1=i+\frac{j-i}{r}, i_2=i+\frac{2(j-i)}{r}, \ldots , i_{r}=j$. \label{line:indices-linear}
\State \label{line:best-guesses-linear}Initialize a list of pairs $(q_1,c_1),\ldots(q_r,c_r)$ each to $(\emptyset,0)$. \Comment{cumulative best guesses and confidences}
\For {$\ell$ iterations}
    \State Let $\pi_1\ldots \pi_r$ be a random permutation of $[r]$.
    \For {$a\in [r]$}
        \State Compute $(\hat{q},\hat{c})\gets \estA(i_{\pi_{a}-1},i_{\pi_a})$. 
        \If{$\hat{q}=q_{\pi_a}$} \Comment{update the confidence on $\estA(i_{\pi_a},i_{\pi_j+1},q_{\pi_a-1})$}
            \State $c_{\pi_a} \gets c_{\pi_a}+\hat{c}$
        \Else
            \State $c_{\pi_a} \gets c_{\pi_a}-\hat{c}$
            \If {$c_{\pi_a}<0$} \Comment{if the guess changes, flip its confidence}
                \State $q_{\pi_a}\gets \hat{q}$ and $c_{\pi_a}\gets -c_{\pi_a}$.
            \EndIf
        \EndIf 
    \EndFor
\EndFor
\State \Return $q_1+\ldots +q_r,\min(c_1\ldots c_r)/\ell$
\EndIf

\EndFunction
\State
\State Initialize a pair $(q,c)$ to $(\emptyset,0)$.
\For {$\log^2 n$ iterations} \Comment{amplification step} 
    \State Let $(\hat{q},\hat{c}) \gets \estA(0,n)$
    \If{$\hat{q}=q$}
        \State $c \gets c+\hat{c}$
    \Else
        \State $c \gets c-\hat{c}$
        \If {$c<0$} 
            \State $q\gets \hat{q}$ and $c\gets -c$.
        \EndIf
    \EndIf 
\EndFor
\State \textbf{output} $q$. \Comment{output $A(x)$}

\end{algorithmic}
\end{algorithm}

\subsection{Proof of Theorem~\ref{thm:main-linear}}

We begin by proving the stated claims about the communication, time, and space complexities, and then proceed to prove correctness.

\subsubsection{Algorithmic Complexity}

We've already shown by Definition~\ref{def:enc-linear} that the length of the encoding $\enc(x)$ is $m = O_\eps(n^{2 + 4\delta})$. We begin by proving the computational complexity claim.

\begin{lemma}\label{lem:time-linear}
    If $A$ runs in time $t$, Algorithm~\ref{alg:main} runs in time $m \cdot O_{\eps,\delta}\left( 1+ \frac{t}{n^2} \right)$.
\end{lemma}

\begin{proof}
    Notice that each function $g_i(x_i)$ is computed exactly $\frac{\log^2 n \cdot M_n}{n}$ times throughout the protocol. (To see this, one can show inductively that $\frac{1}{j-i}$ fraction of the $\LDC(x)$'s in the computation of $A(i,j)$ are used towards the decoding of each $x_k$, $k \in (i, j]$.) Other than that, one must decode of each of the $\log^2 n \cdot M_n$ copies of $\LDC(x)$ throughout the stream, each of which takes $N(n) + \polylog(n)$ time to read and decode (see Theorem~\ref{thm:ldc}). 
    This gives a total computation time of $\log^2 n \cdot M(n) \cdot \left( N(n) + \polylog(r) + \frac{t}{n} \right) = O_{\eps,\delta}\left( m \cdot \left( 1 + \frac{t}{n^2} \right) \right)$.
\end{proof}

Next, we show that Algorithm~\ref{alg:main-linear} satisfies the required space guarantees. We'll first show that the confidences $c$ can be stored in $O_{\delta}(\log n)$ bits.

\begin{lemma} \label{lem:bits-conf-linear}
    For any $i,j,q$ where $j-i$ is a power of $r$, let $(\hat{q},\hat{c}):=\estA(i,j)$. It holds that $\hat{c}$ can be computed and stored as a fraction $(c_1,c_2)=c_1/c_2$ where $c_1\leq c_2 = (j-i)^{\log_r \ell} \cdot 4Q$. 
\end{lemma}

\begin{proof}
    We show this by induction on $j-i$ for $j-i\leq n$. For the base case where $j-i=1$, this statement holds by Theorem~\ref{thm:ldc} since $c_2 = 4Q$.
    
    Assume the statement holds for all $q'$ and $j'-i'<j-i$ where $j'-i'$ is a power of $r$. In the computation of $\estA(i,j)$, for $a\in r$, each value of $c_a$ is ultimately the sum of $\leq \ell$ values outputted by the function $\estA(i_{a-1},i_a)$. Since $\frac{i_a-i_{a-1}}{r}=\frac{j-i}{r}$, we can use the induction hypothesis to write each value as a fraction with denominator $\left( \frac{j-i}{r} \right)^{\log_r \ell} \cdot 4Q$. Then, each $c_a/\ell$ is a fraction with denominator 
    \[
        \left( \frac{j-i}{r} \right)^{\log_r \ell} \cdot 4Q \cdot \ell 
        = \ell^{\log_r \frac{j-i}{r}} \cdot 4Q \cdot \ell 
        = \ell^{\log_r(j-i)} \cdot 4Q
        = (j-i)^{\log_r \ell} \cdot 4Q,
    \]
    and so the output confidence $\min_{a\in [r]}\{c_a/\ell\}$ can be written with denominator $c_2 = (j-i)^{\log_r \ell} \cdot 4Q$.
\end{proof}

Now, we return to showing that Algorithm~\ref{alg:main-linear} satisfies the required space guarantees. 

\begin{lemma} \label{lemma:space-linear}
    Algorithm~\ref{alg:main-linear} uses $s \cdot O_\eps \left( (\log n)^{O(1/\delta)} \right)$ bits of space. 
\end{lemma}

\begin{proof}
    Define $\alpha = \alpha(n) = O_\eps(\log^{O(1/\delta)} n)$ to be the runtime and thus an upper bound on the space required in the local decoding with confidence of $\LDC(x)$ as given in Theorem~\ref{thm:ldc}. We claim the computation of $\estA(0,n)$ takes $\alpha + s + \log_r n \cdot r \cdot \left(s + O_{\eps,\delta}(\log n) \right)$ bits of storage. We show by induction on $j-i$ (where $j-i$ is a power of $r$ that the computation of $\estA(i,j)$ requires at most $\alpha + s + \log_r (j-i) \cdot r \cdot \left(s + O_{\eps,\delta}(\log n) \right)$ bits of space.
    
    The base case holds because when $j-i=1$ because then only $\alpha + s$ space is required. 

    For the inductive step, the computation of $\estA(i,j)$ requires $\ell$ computations of $\estA(i_{a-1}, i_a)$ for each $a \in [r]$. By the inductive hypothesis, each such computation takes $\alpha + s + \log_r \left( \frac{j-i}{r} \right) \cdot r \cdot \left(s + O_{\eps,\delta}(\log n) \right)$ bits of space. Each such computation is done in series, and between iterations only the pairs $\{ (q_a, c_a) \}_{a \in [r]}$ are tracked. Each $q_a$ is a state of size $s$. Each $c_a$ is the sum of $\le \ell$ confidences that are the output of computation $\estA(i_{a-1}, i_a)$. By Lemma~\ref{lem:bits-conf} the output confidence of $\estA(i_{a-1}, i_a)$ is a fraction $\le 1$ with denominator $\left( \frac{j-i}{r} \right)^{\log_r \ell} \cdot 4Q$, so $c_a$ can be represented in $2\log_2 \left( \ell \cdot \left( \frac{j-i}{r} \right)^{\log_r \ell} \cdot 4Q \right) < O_{\eps,\delta}( \log_2 n )$ bits.
    Moreover, in each chunk, we must compute a permutation of $r$, which takes at most $r \log_2 r = r \cdot o_{\delta}(\log n)$ bits. In total, this takes 
    \begin{align*}
        \le \left[ \alpha + s + \log_r \left( \frac{j-i}{r} \right) \cdot r \cdot \left(s + O_{\eps,\delta}(\log n) \right) \right] + r \cdot \left( s + O_{\eps,\delta}(\log n) \right) \\
        = \alpha + s + \log_r \left( j-i \right) \cdot r \cdot \left(s + O_{\eps,\delta}(\log n) \right)
    \end{align*}
    bits of space, as desired.

    Finally, the amplification part of Algorithm~\ref{alg:main} requires $\log^2 n$ iterations. Between iterations, we are only storing a pair $(q,c)$. This increases the number of bits required by at most $s + O_{\eps,\delta}(\log n)$.

    In total, the total space needed throughout the entire algorithm is 
    \[
        < s \cdot \log^{1/\delta+1} n + O_{\eps,\delta}(\log^{1/\delta+2} n) + \alpha < s \cdot O_\eps \left( (\log n)^{O(1/\delta)} \right)
    \]
    bits.
\end{proof}

\subsubsection{Correctness}

We now show correctness. Formally, correctness is shown by the following lemma.

\begin{lemma} \label{lem:correct-linear}
    When at most $\frac14 - \eps$ fraction of the stream $\enc(x)$ is corrupted, Algorithm~\ref{alg:main-linear} outputs $A(x)$ with probability $\ge 1-\exp{(-\eps^2\log^2n/32)}$. 
\end{lemma} 

We prove the following statement which will easily imply Lemma~\ref{lem:correct-linear}. For a given $i,j$ and associated string that is read by the algorithm in the computation, let the random variable $(q(i,j),c(i,j))$ denote the state and confidence after the computation $\estA(i,j)$. We define the \emph{signed confidence}, denoted $c^\pm(i, j)$, to be defined as $+ c(i, j)$ if $q(i, j) = A_{i,j}(x)$ and $-c(i, j)$ otherwise. For intuition, the more positive $c^\pm(i,j)$ is, the more correct with higher confidence the output $(q(i,j), c(i,j))$ is, whereas the more negative $c_\pm(i,j)$ is, the more $q(i,j)$ is incorrect with high confidence. So, $c^\pm(i,j)$ gives us a scaled measure of how correct the output of $\estA(i,j)$ is.

\begin{lemma} \label{lem:corr-gen-linear}
    For any $0\leq i<j\leq n$ and $e\in \bbR$, given $\le \frac14-\frac\eps4-\frac\eps4\cdot \frac{\log(j-i)}{\log n} - e$ fraction of corruptions in the bits read by the computation of $\estA(i,j)$, we have $\bbE[c^\pm(i,j)] > e$.
\end{lemma}

We defer the proof to Section~\ref{sec:corr-proof-linear} and return to the proof of Lemma~\ref{lem:correct-linear}.

\begin{proof}[Proof of Lemma~\ref{lem:correct-linear}]
    In the special case where $i=0,j=n$ and there are $\le \frac14-\frac\eps2 - e < \frac14 - \frac\eps4 - \frac\eps4 \cdot \frac{\log(j-i)}{r} - e$ errors (where the inequality follows because $\log n < \eps \cdot r$), by Lemma~\ref{lem:corr-gen-linear} we have that $\bbE[c^\pm (0,n)] > e$.

    Then, the final amplification step of the protocol computes the pair $(q(0,n),c(0,n))$ for $\log^2 n$ times, where in each chunk $i$ we denote the fraction of error to be $\left( \frac14-\frac\eps2 - e_i \right)$, where $\sum_i e_i \geq \frac{\eps \log^2 n}{2}$ since we assumed the total error was $\le \frac14 - \eps$. Also, the protocol outputs $A(x)$ if and only if $\sum_i c^\pm_i(0,n)$ (denoting the value of $c^\pm(0,n)$ in the $i$'th chunk) is positive. By Azuma's inequality, 
    \[
        Pr\left[\sum_i c^\pm_i(0,n) > 0 \right] \geq Pr\left[\sum_i \left( c^\pm_i (0,n)-e_i \right) >-\frac{\eps \log^2n}{2} \right] \geq 1 - \exp \left( - \frac{\eps^2 \log^2 n}{32} \right).
    \]
\end{proof}

Combining Lemmas~\ref{lem:time-linear},~\ref{lemma:space-linear}, and~\ref{lem:correct-linear}, this concludes the proof of Theorem~\ref{thm:main-linear}.


\subsubsection{Proof of Lemma~\ref{lem:corr-gen-linear}} \label{sec:corr-proof-linear}

We prove Lemma~\ref{lem:corr-gen-linear} by induction on $j-i$. We'll show the statement assuming it holds for $j'-i'<j-i$ (where $j'-i'$ is a power of $r$).

We begin with the base case of $j-i=1$. In this case, the algorithm simply reads (a corrupted version of) $\LDC(x)$, and by Theorem~\ref{thm:ldc}, it holds that $\bbE[c^\pm(i,j)] > e$ since there is at most $\left( \frac14 - \frac\eps4 - e \right)$ error.

We start with some useful notation. In the computation of $\estA(i,j)$, we say that the output of the computation of $\estA(i_{a-1}, i_a)$ in the $b$'th chunk is $(\hat{q}_{a,b}, \hat{c}_{a,b})$. We define the \emph{signed confidence} $\hat{c}^\pm_{a,b}$ to be the confidence $\hat{c}_{a,b}$ signed by whether the guess $\hat{q}_{a,b}$ is correct: that is,
\[
    \hat{c}^\pm_{a,b} := \begin{cases}
        \hat{c}_{a,b} & \text{$\hat{q}_{a,b} = A_{i_{a-1},i_a}(x)$} \\
        -\hat{c}_{a,b} & \text{$\hat{q}_{a,b} \neq A_{i_{a-1},i_a}(x)$}.
    \end{cases}
\]

The cumulative best guess and confidence for $A_{i_{a-1},i_a}(x)$ after all $\ell$ chunks is denoted by $(q_{a}, c_{a})$. We let $q_{0} = \emptyset$. We define a \emph{cumulative signed confidence} $c^\pm_{a}$ as follows:
\[
    c^\pm_{a} := \begin{cases}
        c_{a} & \text{$q_{a} = A_{i_{a-1},i_a}(x)$} \\
        -c_{a} & \text{$q_{a} \neq A_{i_{a-1},i_a}(x)$}.
    \end{cases}
\]

We show a lower bound on the cumulative signed confidence at the end of all $\ell$ chunks $c^\pm_{a}$ that holds with high probability.

\begin{lemma} \label{lem:cpm-calc-linear}
    Let $a\in [r]$ index the $r$ intervals $(i_{a-1}, i_a]$ in the computation of $\estA(i,j)$. Let each of the $\ell$ chunks in the computation be indexed by $b \in [\ell]$ and suppose chunk $b$ has $\left( \frac14-\frac\eps4-\frac\eps4\cdot \frac{\log \frac{j-i}{r}}{\log n} - e_b \right) $ fraction of corruption. Then, for all $a \in [r]$, with probability $\ge 1-\exp\left( - \frac{\log^2 n}{8} \right)$ it holds that $c^\pm_{a} > \sum_{b\leq \ell} e_b - \frac{\ell}{\log n}$. 
\end{lemma}

\begin{proof}
    We first show that for all $a$, it holds that $c^\pm_{a} \geq \sum_{b\leq \ell} \hat{c}^\pm_{a,b}$. To see this, define $c^\pm_{a,b}$ to be the value of $c_a$ at the end of chunk $b$ if $q_a$ is correct at the end of chunk $b$, and otherwise, $-c_a$, so that $c^\pm_a=c^\pm_{a,\ell}$. Then, $c^\pm_{a,b}-c^\pm_{a,b-1} \geq \hat{c}^\pm_{a,b}$, which implies that $c^\pm_{a} \geq \sum_{b\leq \ell} \hat{c}^\pm_{a,b}$. 

    Consider the sequence $\left\{ X_{b} := \sum_{d\leq b} \left( \hat{c}^\pm_{a,d} - e_d \right) \right\}_{b \in [\ell]}$. Because $\bbE \left[ \hat{c}^{\pm}_{a,b} \right] > e_{b}$ by induction on induction on $j-i$, the sequence $\{ X_b \}_{b \in [\ell]}$ is a submartingale. It also holds that $|\hat{c}^\pm_{a,b} - e_b| \le 2$.

    Then, by Azuma's inequality, it holds that 
    \begin{align*}
        \Pr \left[ c^\pm_{a} > \sum_{b\leq \ell} e_b - \frac{\ell}{\log n} \right] \geq Pr \left[ X_{\ell} > -\frac{\ell}{\log n} \right] > 1- \exp \left( - \frac{\log^2 n}{8} \right).
    \end{align*}
\end{proof}

Now, we need a way to reconcile the values of $e$ and $e_b$. Recall that from the statement of Lemma~\ref{lem:corr-gen-linear} that $e$ is defined to be such that there is $\left( \frac14 - \frac\eps4 - \frac\eps4 \cdot \frac{\log{j-i}}{\log n} - e \right)$ fraction of error in the bits read by $\estA(i,j)$. In Lemma~\ref{lem:cpm-calc-linear}, chunk $b$ had $\left( \frac14 - \frac\eps4 - \frac\eps4 \cdot \frac{\log \frac{j-i}{r}}{\log n} - e_b \right)$ fraction of error. This means that the total fraction of error is equal to 
\begin{align*}
    \frac14 - \frac\eps4 - \frac\eps4 \cdot \frac{\log(j-i)}{\log n} - e
    &= \frac1\ell \sum_{b \le \ell} \left( \frac14 - \frac\eps4 - \frac\eps4 \cdot \frac{\log \frac{j-i}{r}}{\log n} - e_b \right) \\
    \implies e &= \frac1\ell \cdot \sum_{b \le \ell} e_b - \frac\eps4 \cdot \frac{\log r}{\log n}.
\end{align*}

Next, we show that the cumulative signed confidence is larger than proportional to $e$ at the end of each of the $\ell$ chunks.

\begin{lemma} \label{lem:final-cpm-linear}
    With probability at least $1-\frac1n$, it holds that for all $a\in [r]$ that $\frac{c^\pm_{a}}{\ell} > e+\frac{1}{\log n}$.
\end{lemma}

\begin{proof}
    Using Lemma~\ref{lem:cpm-calc-linear} with a union bound, it holds with probability at least 
    \begin{align*}
        1 - \sum_{1 \le a \le r} \exp \left( - \frac{\log^2 n}{8} \right)
        \ge 1 - r \cdot \exp \left( - \frac{\log^2 n}{8} \right) 
        \ge 1-\frac1n
    \end{align*}
    that
    \begin{align*}
        \frac{c^\pm_{a}}{\ell} &> \frac{\sum_{b\leq \ell} e_b - \frac{\ell}{\log n}}{\ell}
        \geq e+ \frac\eps4 \cdot \frac{\log r}{\log n} - \frac{1}{\log n} 
        \geq e+ \frac{1}{\log n}
    \end{align*}
    since for sufficiently large $n$ it holds that $\frac{\eps}{4} \cdot \log r > 2$.
\end{proof}

We return to the computation of $\bbE[c^\pm(i,j)]$. Note that 
\[
    c^\pm(i,j)= \begin{cases}
        +\min_{a\in [r]}\{\frac{c_{a}}\ell\} & q(i,j) = A_{i,j}(x) \\
        -\min_{a\in [r]}\{\frac{c_{a}}\ell\} & q(i,j) \neq A_{i,j}(x).
    \end{cases}
\]

Define the event $E$ to be the event where $\frac{c^\pm_{a}}{\ell} > e + \frac1{\log n}$ for all $a \in [r]$. By Lemma~\ref{lem:final-cpm-linear}, $E$ holds with probability at least $1-\frac1n$. We split the proof of Lemma~\ref{lem:corr-gen-linear} based on whether $e + \frac{1}{\log n} \ge 0$. 

If $e+\frac{1}{\log n} \geq 0$, then whenever $E$ holds we have that for all $a$ we have $c^\pm_{a}>0$ and thus $q(i,j) = A_{i,j}(x)$. Therefore,
\begin{align*}
    \bbE[c^\pm(i,j)] 
    &\geq \bbE \left[ \min_{a \in [r]} \left\{ \frac{c_{a}}{\ell} \right\} \Bigg| E \right] \cdot \Pr [E]
    + (\text{min value of $c^\pm(i,j)$}) \cdot \left( 1-\Pr [E] \right) \\ 
    &> \left(e+\frac{1}{\log n}\right)\left(1-\frac1n\right) +(-1)\cdot \frac{1}{n} \\
    &= e + \frac{n -e\log{n} - 1 - \log n}{n \log n} \\
    &> e.
\end{align*}

If $e+\frac{1}{\log n} < 0$, then if $q(i,j) = A_{i,j}(x)$, then $\bbE[c^\pm(i,j)]\geq 0 > e+\frac{1}{\log n} \geq e$. Otherwise, there is at least one value of $a$ for which $q_a$ is incorrect. Conditioned on $E$, it holds that $\frac{c^\pm_{a}}{\ell} 
> e+\frac{1}{\log n}$. Since $q_{a}$ is incorrect it holds that $c^\pm_{a}=-c_{a}$, implying that $\frac{c_{a}}{\ell} < -e-\frac{1}{\log n}$. As such,
\begin{align*}
    \bbE[c^\pm(i,j)] 
    &\ge \bbE\left[-\min_{a\in [r]}\left\{\frac{c_{a}}\ell\right\} \Bigg| E \right] \cdot \Pr [E] + (\text{min value of $c^\pm(i,j)$}) \cdot \left( 1-\Pr [E] \right) \\
    &> \left(e+\frac{1}{\log n}\right)\left(1-\frac1n\right) +(-1)\cdot \frac{1}{n} \\
    &= e + \frac{n -e\log{n} - 1 - \log n}{n \log n} \\
    &> e.
\end{align*}

This concludes the proof of Lemma~\ref{lem:corr-gen-linear}.
\section{Noise Resilient Streaming for General Algorithms} \label{sec:general-streaming}






We now consider general streaming algorithms, whose computational may be sequential in nature. Our main result is a noise resilient conversion for deterministic streaming algorithms $A$. Compared to our scheme for linear algorithms from Section~\ref{sec:linear}, the length of our encoding is larger: $n^{4+\delta}$ compared to $n^{2+\delta}$.


\mainthm*

Throughout this section, we use the notation
\[
    A(q, \hat{x}) \in \{ 0, 1 \}^s
\]
to denote the state of algorithm $A$ when starting with the state $q \in \{ 0, 1 \}^s$ and executing on $\hat{x}$ received in a stream. By definition, there is an explicit algorithm that computes $A(q, \hat{x})$ in $s$ space. We also use the shorthand $A(\hat{x})$ to denote $A(\emptyset, \hat{x})$. Notice that $A(x) = A(\emptyset, x)$ is simply the output state of the stream.

\subsection{Statement of Algorithm}

Let $\eps > 0$. Throughout the section let $r = \log^{1/\delta} n$ and define $\ell := r^2\log^4{n}$. We also assume for simplicity that $\log_r n$ is an integer. We will also assume that $n > \max \left\{ \exp(\exp({8\delta/\eps})), \exp(1/\eps) \right\}$ is sufficiently large, so that also $\ell < n$.

We begin by specifying the encoding $\enc(x)$ that Alice sends in the stream.

\begin{definition}[$\enc(x)$] \label{def:enc}
    Let $\LDC(x)$ be a locally decodable code with $N(n) = |\LDC(x)| = O_\eps(n^{1 + \delta})$ and query complexity $Q = O_\eps(\log^{2/\delta} n)$ satisfying the guarantees of Theorem~\ref{thm:ldc} for $\eps/8$. Then Alice sends $\enc(x):=\LDC(x)^{\log^2 n \cdot M_n}$ where $M_n = (r \cdot \ell)^{\log_r n}$ (that is, $\enc(n)$ is $\log^2 n \cdot M_n$ copies of $\LDC(x)$). In particular, $m = |\enc(x)| = (r \cdot \ell)^{\log_r n} \cdot \log^2 n \cdot N(n) = O_\eps(n^{4+6\delta})$.
\end{definition}


Next, we describe Bob's streaming algorithm $B_A:=B(A)$. Before stating it formally, we provide a high-level description.

\paragraph{Description:} 
Let $z$ be the incoming stream that Bob receives. At a high level, the algorithm $B_A$ runs by generating many guesses for $A(x)$ along with confidences indicating how likely that guess was correct (depending on how much error is witnessed while generating the guess). At the end, the guess that has the most cumulative confidence is outputted.

For $i,j$ where $j-i$ is an integer power of $r$, we define the algorithm $\estA(i, j, q)$ that has the following syntax:
\begin{itemize}
\item 
    $\estA(i, j, q)$ takes as input two indices $i, j \in [n]$ along with a state $q \in \{ 0, 1 \}^s$ which represents a guess for the state of $A(q, x(i:j])$. 
\item 
    It reads $N_{j-i} := (r \cdot \ell)^{\log_r(j-i)} \cdot N(n)$ many bits of $z$ after which it outputs a guess for $A(q, x(i:j])$ along with a confidence. 
\end{itemize}
In particular, $\estA(0, n, q)$ outputs a guess for $A(x)$ along with a confidence. By running $\estA(0, n, q)$ many times and aggregating the guesses weighted by the confidences, we obtain a final output.

At a high level, to compute $\estA(i, j, q)$, we break down the interval $(i : j]$ into $r$ subintervals $(i_0 : i_1], (i_1 : i_2], \dots, (i_{r-1} : i_r]$ where $i_a = i + a \cdot \frac{j-i}{r}$. We then recursively formulate guesses for each $A \left( x\left(1:i_a\right] \right)$ over many iterations by calling $\estA \left( i_{a-1}, i_a, q_{a-1} \right)$ for some state $q_{a-1}$. The choice of $a$ in each iteration must be randomized, so that the adversary cannot attack any single choice of $a$. More specifically, we will split up the $N_{j-i}$ length input stream into $\ell$ chunks, each of length $N_{j-i} / \ell$ bits. Each chunk is split into $r$ sections, each of size $N_{j-i} / (\ell \cdot r) = N_{(j-i)/r}$. In each chunk, we pick a random permutation $(\pi_1, \pi_2, \dots, \pi_r)$ of $[r]$. In the $a$'th section of the chunk, we will compute on the subproblem corresponding to interval $\left( i_{\pi_a-1} : i_{\pi_a} \right]$ by calling $\estA \left( i_{\pi_a-1}, i_{\pi_a}, ~ q_{\pi_a-1} \right)$. 

Note that as the algorithm $A$ is computed sequentially, computing a guess for $A \left( x \left[1:i_a\right] \right)$ requires having a starting state $q_{a-1}$ for $A \left( x\left[1:i_{a-1}\right] \right)$. Thus, throughout the entire computation, we will keep track of the best guess and associated confidence for each of the $r$ states $A\left(x \left[ 1 : i_a \right] \right)$, denoted $(q_a, c_a)$ and all initialized to $(\emptyset, 0)$. 

These pairs $(q_a, c_a)$ are maintained as follows. Whenever a recursive call to $\estA \left( i_{a-1}, i_a, q_{a-1} \right)$ is completed, outputting $(\hat{q}, \hat{c})$, then: if $q_a = \hat{q}$ then the confidence $c_a$ is increased by $\hat{c}$, and if $q_a \not= \hat{q}$ then the confidence $c_a$ is decreased by $\hat{c}$. If the confidence $c_a$ becomes negative, then we are more sure that the correct state is $\hat{q}$ rather than $q_a$: then $q_a$ is replaced by $\hat{q}$ and the confidence is negated (so that it's positive). In this last case where $q_a$ is replaced, we have no more reason to believe that further computations of $q_{a'}$, $a' > a$, are correct since they all depended on $q_a$, so we erase all such pairs $(q_{a'}, c_{a'})$ and reset them to $(\emptyset, 0)$. A key point in our analysis is to understand why this does not cause the error probability to accumulate.

\begin{algorithm}[H]
\caption{Bob's Noise Resilient Algorithm $B_A$}
\label{alg:main}
\renewcommand{\thealgorithm}{}
\floatname{algorithm}{}

\begin{algorithmic}[1]
\State \textbf{input} $n\in \bbN$ and stream $z \in \{ 0, 1 \}^m$.
\Function{$\estA$}{$i,j,q$} \Comment{compute the state of $A$ starting at state $q$ between steps $i$ and $j$}

\If{$j=i+1$}
\State Read the next $|\LDC(x)|$ bits of the stream $y$. \label{line:ldc}
\State Using Theorem~\ref{thm:ldc} compute guess $\hat{b}$ for $x[j]$ and confidence $\hat{c}$ in $s \cdot O_\eps\left( (\log n)^{O(1/\delta)} \right)$ bits of space.
\State \Return $A(q,\hat{b}), c$.

\Else 
\State Let $i_0=i, i_1=i+\frac{j-i}{r}, i_2=i+\frac{2(j-i)}{r}, \ldots , i_{r}=j$. \label{line:indices}
\State \label{line:best-guesses} Initialize a list of pairs $(q_1,c_1),\ldots(q_r,c_r)$ each to $(\emptyset,0)$. \Comment{cumulative best guesses and confidences}
\For {$\ell$ iterations}
    \State Let $\pi_1\ldots \pi_r$ be a random permutation of $[r]$.
    \State Set $(q'_1,c'_1)=(q_1,c_1),\ldots(q'_r,c'_r)=(q_r,c_r)$.
    \For {$a\in [r]$}
        \State Compute $(\hat{q},\hat{c})\gets \estA(i_{\pi_{a}-1},i_{\pi_a},q'_{\pi_a-1})$ where $q'_0:=q$. 
        \If{$\hat{q}=q'_{\pi_a}$} \Comment{update the confidence on $\estA(i_{\pi_a},i_{\pi_j+1},q_{\pi_a-1})$}
            \State $c_{\pi_a} \gets c'_{\pi_a}+\hat{c}$
        \Else
            \State $c_{\pi_a} \gets c'_{\pi_a}-\hat{c}$
            \If {$c_{\pi_a}<0$} \Comment{if the guess changes, reset its confidence}
                \State $q_{\pi_a}\gets \hat{q}$ and $c_{\pi_a}\gets -c_{\pi_a}$.
            \EndIf
        \EndIf 
    \EndFor
    \If{some $q_a$ was changed from its value at the beginning of the chunk}
        \State{For all $i > a$, set $(q_i, c_i) \gets (\emptyset, 0)$.} 
        \Comment{reset state and confidence for states with $i>a$}
    \EndIf
\EndFor
\State \Return $q_r,\min(c_1\ldots c_r)/\ell$
\EndIf

\EndFunction
\State
\State Initialize a pair $(q,c)$ to $(\emptyset,0)$.
\For {$\log^2 n$ iterations} \Comment{amplification step} 
    \State Let $(\hat{q},\hat{c}) \gets \estA(0,n,\emptyset)$
    \If{$\hat{q}=q$}
        \State $c \gets c+\hat{c}$
    \Else
        \State $c \gets c-\hat{c}$
        \If {$c<0$} 
            \State $q\gets \hat{q}$ and $c\gets -c$.
        \EndIf
    \EndIf 
\EndFor
\State \textbf{output} $q$. \Comment{output $A(x)$}

\end{algorithmic}
\end{algorithm}


\subsection{Proof of Theorem~\ref{thm:main}}



The proofs that the algorithm satisfies the required space, computational and communication guarantees are essentially the same as in Section~\ref{sec:linear}, but we reproduce them for completeness. The proof of correctness is somewhat more complicated but also follows the same general outline.

\subsubsection{Algorithmic Complexity}

We begin by proving the stated claims about communication, time, and space complexities. We've already shown by Definition~\ref{def:enc} that the length of the encoding $\enc(x)$ is $m = O_\eps(n^{4 + 6\delta})$. We thus proceed by proving the computational complexity claim.

\begin{lemma}\label{lem:time}
    If $A$ runs in time $t$, Algorithm~\ref{alg:main} runs in time $m \cdot O_{\eps,\delta}\left( 1+ \frac{t}{n^2} \right)$.
\end{lemma}

\begin{proof}
    Notice that each step $A(q_{i-1}, x[i])$ is computed exactly $\frac{\log^2 n \cdot M_n}{n}$ times throughout the protocol. (To see this, one can show inductively that $\frac{1}{j-i}$ fraction of the $\LDC(x)$'s in the computation of $A(i,j,q)$ are used towards the decoding of $x_k$ for each $k \in (i, j]$.) Other than that, one must decode of each of the $\log^2 n \cdot M_n$ copies of $\LDC(x)$ throughout the stream, each of which takes $N(n) + \polylog(n)$ time to read and decode (see Theorem~\ref{thm:ldc}). 
    This gives a total computation time of $\log^2 n \cdot M(n) \cdot \left( N(n) + \polylog(r) + \frac{t}{n} \right) = O_{\eps,\delta}\left( m \cdot \left( 1 + \frac{t}{n^2} \right) \right)$.
\end{proof}

Next, we show that Algorithm~\ref{alg:main} satisfies the required space guarantees. We'll first show that the confidences $c$ can be stored in $O_{\delta}(\log n)$ bits.

\begin{lemma} \label{lem:bits-conf}
    For any $i,j,q$ where $j-i$ is a power of $r$, let $(\hat{q},\hat{c}):=\estA(i,j,q)$. It holds that $\hat{c}$ can be computed and stored as a fraction $(c_1,c_2)=c_1/c_2$ where $c_1\leq c_2 = (j-i)^{\log_r \ell} \cdot 4Q$. 
\end{lemma}

\begin{proof}
    We show this by induction on $j-i$ for $j-i\leq n$. For the base case where $j-i=1$, this statement holds by Theorem~\ref{thm:ldc} since $c_2 = 4Q$.
    
    Assume the statement holds for all $q'$ and $j'-i'<j-i$ where $j'-i'$ is a power of $r$. In the computation of $\estA(i,j,q)$, for $a\in r$, each value of $c_a$ is ultimately the sum of $\leq \ell$ values outputted by the function $\estA(i_{a-1},i_a)$. Since $\frac{i_a-i_{a-1}}{r}=\frac{j-i}{r}$, we can use the induction hypothesis to write each value as a fraction with denominator $\left( \frac{j-i}{r} \right)^{\log_r \ell} \cdot 4Q$. Then, each $c_a/\ell$ is a fraction with denominator 
    \[
        \left( \frac{j-i}{r} \right)^{\log_r \ell} \cdot 4Q \cdot \ell 
        = \ell^{\log_r \frac{j-i}{r}} \cdot 4Q \cdot \ell 
        = \ell^{\log_r(j-i)} \cdot 4Q
        = (j-i)^{\log_r \ell} \cdot 4Q,
    \]
    and so the output confidence $\min_{a\in [r]}\{c_a/\ell\}$ can be written with denominator $c_2 = (j-i)^{\log_r \ell} \cdot 4Q$.
\end{proof}

Now, we return to showing that Algorithm~\ref{alg:main} satisfies the required space guarantees. 

\begin{lemma} \label{lemma:space}
    Algorithm~\ref{alg:main} uses $s \cdot O_\eps \left( (\log n)^{O(1/\delta)} \right)$ bits of space. 
\end{lemma}

\begin{proof}
    Define $\alpha = \alpha(n) = O_\eps(\log^{O(1/\delta)} n)$ to be the runtime and thus an upper bound on the space required in the local decoding with confidence of $\LDC(x)$ as given in Theorem~\ref{thm:ldc}. We claim the computation of $\estA(0,n,\emptyset)$ takes $\alpha + s + \log_r n \cdot r \cdot \left(s + O_{\eps,\delta}(\log n) \right)$ bits of storage. We show by induction on $j-i$ (where $j-i$ is a power of $r$ that the computation of $\estA(i,j,q)$ requires at most $\alpha + s + \log_r (j-i) \cdot r \cdot \left(s + O_{\eps,\delta}(\log n) \right)$ bits of space.
    
    
    The base case holds because when $j-i=1$ because then only $\alpha + s$ space is required. 

    For the inductive step, the computation of $\estA(i,j,q)$ requires $\ell$ computations of $\estA(i_{a-1}, i_a, q_{a-1})$ for each $a \in [r]$. By the inductive hypothesis, each such computation takes $\alpha + s + \log_r \left( \frac{j-i}{r} \right) \cdot r \cdot \left(s + O_{\eps,\delta}(\log n) \right)$ bits of space. Each such computation is done in series, and between iterations only the pairs $\{ (q_a, c_a) \}_{a \in [r]}$ are tracked. Each $q_a$ is a state of size $s$. Each $c_a$ is the sum of $\le \ell$ confidences that are the output of computation $\estA(i_{a-1}, i_a, q_{a-1})$. By Lemma~\ref{lem:bits-conf} the output confidence of $\estA(i_{a-1}, i_a, q_{a-1})$ is a fraction $\le 1$ with denominator $\left( \frac{j-i}{r} \right)^{\log_r \ell} \cdot 4Q$, so $c_a$ can be represented in $2\log_2 \left( \ell \cdot \left( \frac{j-i}{r} \right)^{\log_r \ell} \cdot 4Q \right) < O_{\eps,\delta}( \log_2 n )$ bits.
    Moreover, in each chunk, we must compute a permutation of $r$, which takes at most $r \log_2 r = r \cdot o_{\delta}(\log n)$ bits. In total, this takes 
    \begin{align*}
        \le \left[ \alpha + s + \log_r \left( \frac{j-i}{r} \right) \cdot r \cdot \left(s + O_{\eps,\delta}(\log n) \right) \right] + r \cdot \left( s + O_{\eps,\delta}(\log n) \right) \\
        = \alpha + s + \log_r \left( j-i \right) \cdot r \cdot \left(s + O_{\eps,\delta}(\log n) \right)
    \end{align*}
    bits of space, as desired.

    Finally, the amplification part of Algorithm~\ref{alg:main} requires $\log^2 n$ iterations. Between iterations, we are only storing a pair $(q,c)$. This increases the number of bits required by at most $s + O_{\eps,\delta}(\log n)$.

    In total, the total space needed throughout the entire algorithm is 
    \[
        < s \cdot \log^{1/\delta+1} n + O_{\eps,\delta}(\log^{1/\delta+2} n) + \alpha < s \cdot O_\eps \left( (\log n)^{O(1/\delta)} \right)
    \]
    bits.
\end{proof}

\subsubsection{Correctness}

Next, we show correctness. Formally, correctness is shown by the following lemma.

\begin{lemma} \label{lem:correct}
    When at most $\frac14 - \eps$ fraction of the stream $\enc(x)$ is corrupted, Algorithm~\ref{alg:main} outputs $A(x)$ with probability $\ge 1-\exp{(-\eps^2\log^2n/32)}$. 
\end{lemma} 

We prove the following statement which will easily imply Lemma~\ref{lem:correct}. For a given $i,j,q$ and associated string that is read by the algorithm in the computation, let the random variable $(q(i,j,q),c(i,j,q))$ denote the state and confidence after the computation $\estA(i,j,q)$. We define the \emph{signed confidence}, denoted $c^\pm(i, j, q)$, to be defined as $+ c(i, j, q)$ if $q(i, j, q) = A(q, x(i:j])$ and $-c(i, j, q)$ otherwise. For intuition, the more positive $c^\pm(i,j,q)$ is, the more correct with higher confidence the output $(q(i,j,q), c(i,j,q))$ is, whereas the more negative $c_\pm(i,j,q)$ is, the more $q(i,j,q)$ is incorrect with high confidence. So, $c^\pm(i,j,q)$ gives us a scaled measure of how correct the output of $\estA(i,j,q)$ is.

\begin{lemma} \label{lem:corr-gen}
    For any $0\leq i<j\leq n$, $q\in \{0,1\}^s$, and $e\in \bbR$, given $\le \frac14-\frac\eps4-\frac\eps4\cdot \frac{\log(j-i)}{\log n} - e$ fraction of corruptions in the bits read by the computation of $\estA(i,j,q)$, we have $\bbE[c^\pm(i,j,q)] > e$.
\end{lemma}

We defer the proof to Section~\ref{sec:corr-proof} and return to the proof of Lemma~\ref{lem:correct}.

\begin{proof}[Proof of Lemma~\ref{lem:correct}]
    In the special case where $i=0,j=n,q=\emptyset$ and there are $\le \frac14-\frac\eps2 - e < \frac14 - \frac\eps4 - \frac\eps4 \cdot \frac{\log(j-i)}{r} - e$ errors (where the inequality follows because $\log n < \eps \cdot r$), by Lemma~\ref{lem:corr-gen} we have that $\bbE[c^\pm (0,n,\emptyset)] > e$.

    Then, the final amplification step of the protocol computes the pair $(q(0,n,\emptyset),c(0,n,\emptyset))$ for $\log^2 n$ times, where in each chunk $i$ we denote the fraction of error to be $\left( \frac14-\frac\eps2 - e_i \right)$, where $\sum_i e_i \geq \frac{\eps \log^2 n}{2}$ since we assumed the total error was $\le \frac14 - \eps$. Also, the protocol outputs $A(x)$ if and only if $\sum_i c^\pm_i(0,n,\emptyset)$ (denoting the value of $c^\pm(0,n,\emptyset)$ in the $i$'th chunk) is positive. By Azuma's inequality, 
    \[
        Pr\left[\sum_i c^\pm_i(0,n,\emptyset) > 0 \right] \geq Pr\left[\sum_i \left( c^\pm_i (0,n,\emptyset)-e_i \right) >-\frac{\eps \log^2n}{2} \right] \geq 1 - \exp \left( - \frac{\eps^2 \log^2 n}{32} \right).
    \]
\end{proof}

Combining Lemmas~\ref{lem:time},~\ref{lemma:space}, and~\ref{lem:correct}, this concludes the proof of Theorem~\ref{thm:main}.


\subsubsection{Proof of Lemma~\ref{lem:corr-gen}} \label{sec:corr-proof}

We remark that this section constitutes the main difference in the analysis of Algorithm~\ref{alg:main-linear} and Algorithm~\ref{alg:main}. In Algorithm~\ref{alg:main-linear}, the recursive calls to $\estA(i_{a-1},i_a)$ can be computed in any order, and so $c^\pm(i,j)$ is easy to analyze and converges quickly. In Algorithm~\ref{alg:main}, the recursive call to $\estA(i_{a-1},i_a,q_{a-1})$ has no hope of providing any useful information until $q_{a-1}$ is correct, and so in this algorithm, the quantity $c^\pm(i,j,q)$ requires a more careful analysis, and ultimately converges slower.

We prove Lemma~\ref{lem:corr-gen} by induction on $j-i$. We'll show the statement for a specific value of $(i,j,q)$ where $j-i$ is a power of $r$ assuming it holds for any $(i',j',q')$ where $j'-i'<j-i$ (and $j'-i'$ is a power of $r$).

We begin with the base case of $j-i=1$. In this case, the algorithm simply reads (a corrupted version of) $\LDC(x)$, and by Theorem~\ref{thm:ldc}, it holds that $\bbE[c^\pm(i,j,q)] > e$ since there is at most $\left( \frac14 - \frac\eps4 - e \right)$ error.

Now, suppose we are in the computation of $\estA(i,j,q)$. We start with some useful notation. 

We say that a state $q'_{a}$ is \emph{correct} when it is equal to the state produced by the algorithm $A$ starting at $q$ and executing steps $i$ to $i_a$ (as defined in Line~\ref{line:indices} of the algorithm). That is, it is equal to $A ( q,x (i : i_a])$.
We note that this is different from requiring only $\estA(i_{a-1}, i_a, q_{a-1})$ to be done correctly from the starting state $q_{a-1}$.

In the computation of $\estA(i,j,q)$, we say that the output of the computation of $\estA(i_{a-1}, i_a, q_{a-1})$ in the $b$'th chunk is $(\hat{q}_{a,b}, \hat{c}_{a,b})$. We define the \emph{signed confidence} $\hat{c}^\pm_{a,b}$ to be the confidence $\hat{c}_{a,b}$ signed by whether the guess $\hat{q}_{a,b}$ is correct: that is,
\[
    \hat{c}^\pm_{a,b} := \begin{cases}
        \hat{c}_{a,b} & \text{$\hat{q}_{a,b}$ is correct} \\
        -\hat{c}_{a,b} & \text{$\hat{q}_{a,b}$ is incorrect}.
    \end{cases}
\]

The cumulative best guess and confidence for $A(q, x(i : i_a])$ after chunk $b$, as initialized in Line~\ref{line:best-guesses}, are denoted by $(q_{a,b}, c_{a,b})$. We let $q_{0,b} = \emptyset$. We can define a \emph{cumulative signed confidence} $c^\pm_{a,b}$ as follows:
\[
    c^\pm_{a,b} := \begin{cases}
        c_{a,b} & \text{$q_{a,b}$ is correct} \\
        -c_{a,b} & \text{$q_{a,b}$ is incorrect}.
    \end{cases}
\]

We first show that $c^\pm_{a,b}$ is correct proportional to $\sum_{d \le b} e_d$ assuming that $q_{a-1,b}$ is correct.

\begin{lemma}\label{lem:cpm-calc}
    Let $a\in [r]$ index the $r$ intervals $(i_{a-1}, i_a]$ in the computation of $\estA(i,j,q)$. Let each of the $\ell$ chunks in the computation be indexed by $b \in [\ell]$ and suppose chunk $b$ has $\left( \frac14-\frac\eps4-\frac\eps4\cdot \frac{\log \frac{j-i}{r}}{\log n} - e_b \right) $ fraction of corruption. Then, for all $a \in [r]$, with probability $\ge 1-a\ell^2\exp\left( - \frac{\log^2 n}{32} \right)$ it holds for all $b \in [\ell]$ that either $c^\pm_{a,b} > \sum_{d\leq b} e_d - \frac{a\ell}{r \log n}$ or $q_{a-1,b}$ is incorrect. 
\end{lemma}

\begin{proof}
    We use induction on $a$. Let us prove the statement for $a$ assuming it is true for $a-1$. Our base case is $a=0$, and it holds that $q_{0,b}$ is correct (meaning equal to $q$) for all $b\in [\ell]$. 

    Fix some chunk $b$. For $d \le b$, we define the \emph{altered confidence} $\hat{c}^\alt_{a,d}$ to be the signed confidence, with a caveat for when $q_{a-1,d-1}$ is incorrect, as follows: 
    \[
        \hat{c}^{\alt}_{a,d} := \begin{cases}
            1 & \text{$q_{a-1, d-1}$ is incorrect} \\
            c^\pm_{a,b} & \text{$q_{a-1, d-1}$ is correct}.
        \end{cases}
    \]
    Then, for any $0 \le b' < b$, the sequence $\left\{ X_{b',d} := \sum_{b' < d' \le d} \left( \hat{c}^\alt_{a,d'} - e_{d'} \right) \right\}_{b' < d \le b}$ is a submartingale. To see this, we have that $\bbE \left[ \hat{c}^\alt_{a,d} \right] \ge \bbE \left[ \hat{c}^{\pm}_{a,d} \right] > e_{d}$, where the last inequality was by induction on induction on $j-i$. Furthermore, $|c^\pm_{a,d} - e_d| \le 2$. Then, by Azuma's inequality, 
    \[
        \Pr \left[ X_{b',b} \le -\frac{\ell}{2r\log n} \right] \le \exp \left( - \frac{\ell^2}{32r^2 \log^2 n(b-b')} \right) \le \exp \left( - \frac{\ell}{32r^2 \log^2 n} \right) = \exp \left( - \frac{\log^2 n}{32} \right).
    \]
    recalling that $\ell = r^2\log^4{n}$. By a union bound, 
    \[
        \Pr \left[ X_{b',b} > -\frac{\ell}{2r \log n} ~~\forall~ 0 \le b' < b \right] \ge 1 - \ell \cdot \exp \left( - \frac{\log^2 n}{32} \right). \numberthis \label{eqn:azuma-union-bound1}
    \]
    

    Now, we return to proving the lemma statement. We have to prove that the statement holds simultaneously for all $b \in [\ell]$ with high probability.
    
    By the inductive hypothesis, Equation~\eqref{eqn:azuma-union-bound1}, and a union bound, we have that with probability $\ge 1 - (a-1)\ell^2 \cdot \exp\left( - \frac{\log^2 n}{32} \right) - \ell \cdot \ell \cdot \exp \left( - \frac{\log^2 n}{32} \right) = 1 - a\ell^2 \cdot \exp \left( - \frac{\log^2 n}{32} \right)$, the following two conditions simultaneously hold:
    \begin{enumerate}
        \item \label{item:assumption1} for all $d$ for which $\sum_{d' \le d} e_{d'} \ge \frac{(a-1)\ell}{r}$ it holds that $c^\pm_{a-1,d} > 0$, or in other words $q_{a-1,d}$ is correct, and
        \item \label{item:assumption2} for all $b \in [\ell]$, it holds that $ X_{b',b} > -\frac{\ell}{2r \log n} ~~\forall~ 0 \le b' < b$.
    \end{enumerate}
    Assume for now that we are not within this failure probability. For any chunk $b$, let $b^*$ be the last chunk $\leq b$ in which $q_{a-1,d}$ changed values. If $q_{a-1,b^*}$ is incorrect, then the lemma is proven, because $q_{a-1,b} = q_{a-1,b^*}$ is also incorrect. So assume that $q_{a-1,b^*}$ is correct.

    If $b^*=b$ then $c^\pm_{a,b} = 0$ since $(q_{a,b}, c_{a,b})$ was reset to $(\emptyset, 0)$. Moreover, because $q_{a-1,b-1}$ has to be incorrect, we have that $\sum_{d < b} e_{d} < \frac{(a-1)\ell}{r}$ by assumption~\ref{item:assumption1}. Then, we have that 
    \[
        c^\pm_{a,b} = 0 > \frac{(a-1)\ell}{r \log n} + 1  - \frac{a\ell}{r\log n} \ge
        \sum_{d < b} e_d + e_b - \frac{a\ell}{r \log n}
        = \sum_{d\leq b} e_d - \frac{a\ell}{r \log n},
    \]
    as desired, where the first inequality follows because $\ell = r^2 \log^4 n > r \log n$.
    
    Otherwise, $b^*<b$. Then, since $q_{a-1,b^*}$ is correct, the value of $q_{a-1,b^*-1}$ in chunk $b^*-1$ was incorrect, and so $\sum_{d < b^*} e_{d} \le \frac{(a-1)\ell}{r \log n}$ by assumption~\ref{item:assumption1}.
    Since $q_{a-1, d}$ is correct for all $b^* \le d < b$, we have that $\hat{c}^\alt_{a,d} = \hat{c}^\pm_{a,d}$ for all $b^* < d \le b$. Moreover, $c^\pm_{a,b} \geq \sum_{b^* < d \le b} \hat{c}^\pm_{a,d}$. To see this, notice that for all $b$, it holds that $c^\pm_{a,b}-c^\pm_{a,b-1} \geq \hat{c}^\pm_{a,b}$. Then,
    \[
        c^\pm_{a,b} \geq \sum_{b^* < d \le b} \hat{c}^\pm_{a,d} = \sum_{b^* < d \le b} \hat{c}^\alt_{a,d} = X_{b^*,b} + \sum_{b^* < d \le b} e_{d}.
    \]
    Therefore, by Equation~\eqref{eqn:azuma-union-bound1}, 

    \begin{align*}
        c^\pm_{a,b} &= X_{b^*,b} + \sum_{b^* < d \le b} e_{d} \\
        &> -\frac{\ell}{2r \log n} + \sum_{d\leq b} e_d - e_{b^*} - \sum_{d<b^*} e_d  \\
        &> -\frac{\ell}{2r \log n} + \sum_{d\leq b} e_d - 1 - \frac{(a-1)\ell}{r \log n} \\
        &> \sum_{d\leq b} e_d - \frac{a\ell}{r \log n}.
    \end{align*}
    This holds for all $b \in [\ell]$, as desired.
\end{proof}

We now take $b = \ell$ in Lemma~\ref{lem:cpm-calc}.
First, we need a way to reconcile the values of $e$ and $e_b$. Recall that from the statement of Lemma~\ref{lem:corr-gen} that $e$ is defined to be such that there is $\left( \frac14 - \frac\eps4 - \frac\eps4 \cdot \frac{\log{j-i}}{\log n} - e \right)$ fraction of error in the bits read by $\estA(i,j,q)$. In Lemma~\ref{lem:cpm-calc}, chunk $b$ had $\left( \frac14 - \frac\eps4 - \frac\eps4 \cdot \frac{\log \frac{j-i}{r}}{\log n} - e_b \right)$ fraction of error. This means that the total fraction of error is equal to 
\begin{align*}
    \frac14 - \frac\eps4 - \frac\eps4 \cdot \frac{\log(j-i)}{\log n} - e
    &= \frac1\ell \sum_{b \le \ell} \left( \frac14 - \frac\eps4 - \frac\eps4 \cdot \frac{\log \frac{j-i}{r}}{\log n} - e_b \right) \\
    \implies e &= \frac1\ell \cdot \sum_{b \le \ell} e_b - \frac\eps4 \cdot \frac{\log r}{\log n}.
\end{align*}

\begin{lemma} \label{lem:final-cpm}
    With probability at least $1-\frac1n$, it holds that for all $a\in [r]$ where $q_{a-1,\ell}$ is correct that $\frac{c^\pm_{a,\ell}}{\ell} > e+\frac{1}{\log n}$.
\end{lemma}

\begin{proof}
    By Lemma~\ref{lem:cpm-calc}, using a union bound, it holds with probability at least 
    \begin{align*}
        1 - \sum_{1 \le a \le r}a \ell^2 \exp \left( - \frac{\log^2 n}{32} \right)
        &\ge 1 - r^2 \ell^2 \exp \left( - \frac{\log^2 n}{32} \right) \\
        &\ge 1 - \ell^3 \exp \left( - \frac{\log^2 n}{32} \right) \\
        &\ge 1-\frac1n
    \end{align*}
    (because $\ell<n$) that for all values of $a$ where $q_{a-1,\ell}$ is correct, we have
    \begin{align*}
        \frac{c^\pm_{a,\ell}}{\ell} &> \frac{\sum_{d\leq \ell} e_d - \frac{a\ell}{r \log n}}{\ell}\\
        &\geq e+ \frac\eps4 \cdot \frac{\log r}{\log n} - \frac{1}{\log n} \\
        &\geq e+ \frac{1}{\log n}
    \end{align*}
    since for sufficiently large $n$, we have that $\frac{\eps}{4} \cdot \log r > 2$.
\end{proof}


We return to the computation of $\bbE[c^\pm(i,j,q)]$. Note that 
\[
    c^\pm(i,j,q)= \begin{cases}
        +\min_{a\in [r]}\{\frac{c_{a,\ell}}\ell\} & \text{$q_r$ is correct} \\
        -\min_{a\in [r]}\{\frac{c_{a,\ell}}\ell\} & \text{$q_r$ is incorrect}.
    \end{cases}
\]
We split the proof of Lemma~\ref{lem:corr-gen} based on whether $e + \frac{1}{\log n} \ge 0$. 


If $e+\frac{1}{\log n} \geq 0$, 
define the event $E$ to be the event where $\frac{c^\pm_{a,\ell}}{\ell} > e + \frac1{\log n}$ for all $a \in [r]$. We show that $E$ holds with probability at least $1-\frac{1}{n}$. This is due to Lemma~\ref{lem:final-cpm}, which says that with probability $\ge 1 - \frac1n$, simultaneously for all $a \in [r]$ if $q_{a-1,\ell}$ is correct then $\frac{c^\pm_{a,\ell}}{\ell} > e + \frac1n \ge 0$. Let $E$ be the event that this holds. Then, since $q_{0,\ell} = q$ is correct by definition, it follows that $q_{a,\ell}$ is correct for all $a \in [r]$. Since $q_{a-1,\ell}$ is correct for all $a \in [r]$, it holds that $\frac{c^\pm_{a,\ell}}{\ell} > e + \frac1n$ also.
Therefore,
\begin{align*}
    \bbE[c^\pm(i,j,q)] 
    &\geq \bbE \left[ \min_{a \in [r]} \left\{ \frac{c_{a,\ell}}{\ell} \right\} \Bigg| E \right] \cdot \Pr [E]
    + (\text{min value of $c^\pm(i,j,q)$}) \cdot \left( 1-\Pr [E] \right) \\ 
    &> \left(e+\frac{1}{\log n}\right)\left(1-\frac1n\right) +(-1)\cdot \frac{1}{n} \\
    &= e + \frac{n -e\log{n} - 1 - \log n}{n \log n} \\
    &> e.
\end{align*}

If $e+\frac{1}{\log n} < 0$, then if $q_{r,\ell}$ is correct, then $\bbE[c^\pm(i,j,q)]\geq 0 > e+\frac{1}{\log n} \geq e$. Otherwise, $q_{r,\ell}$ is incorrect. By Lemma~\ref{lem:final-cpm}, with probability $\ge 1 - \frac1n$, it holds that $\frac{c^\pm_{a,\ell}}{\ell} > e + \frac1{\log n}$ for all $a$ as long as $q_{a-1,\ell}$ is correct. Let the event this occurs be $E$. If $q_{r,\ell}$ is incorrect, there is a smallest $a > 0$ such that $q_{a,\ell}$ is incorrect (note that $q_{0, \ell} = 0$ is always correct), for which it holds that $\frac{c^\pm_{a,\ell}}{\ell} 
> e+\frac{1}{\log n}$ given $E$. Since $q_{a,\ell}$ is incorrect it holds that $c^\pm_{a,\ell}=-c_{a,\ell}$, implying that $\frac{c_{a,\ell}}{\ell} < -e-\frac{1}{\log n}$. As such,
\begin{align*}
    \bbE[c^\pm(i,j,q)] 
    &\ge \bbE\left[-\min_{a\in [r]}\left\{\frac{c_{a,\ell}}\ell\right\} \Bigg| E \right] \cdot \Pr [E] + (\text{min value of $c^\pm(i,j,q)$}) \cdot \left( 1-\Pr [E] \right) \\
    &> \left(e+\frac{1}{\log n}\right)\left(1-\frac1n\right) +(-1)\cdot \frac{1}{n} \\
    &= e + \frac{n -e\log{n} - 1 - \log n}{n \log n} \\
    &> e.
\end{align*}

This concludes the proof of Lemma~\ref{lem:corr-gen}.



\bibliographystyle{alpha}
\bibliography{refs}

\newcommand{\etalchar}[1]{$^{#1}$}
\begin{thebibliography}{BEJWY22}

\bibitem[AGM12]{AhnGM12}
Kook~Jin Ahn, Sudipto Guha, and Andrew McGregor.
\newblock {\em Analyzing Graph Structure via Linear Measurements}, pages
  459--467.
\newblock 2012.

\bibitem[AMS96]{AlonMS96}
Noga Alon, Yossi Matias, and Mario Szegedy.
\newblock The space complexity of approximating the frequency moments.
\newblock In {\em Proceedings of the twenty-eighth annual ACM symposium on
  Theory of computing}, pages 20--29, 1996.

\bibitem[BE14]{BravermanE14}
Mark Braverman and Klim Efremenko.
\newblock {List and Unique Coding for Interactive Communication in the Presence
  of Adversarial Noise}.
\newblock In {\em 2014 IEEE 55th Annual Symposium on Foundations of Computer
  Science (FOCS)}, pages 236--245, Los Alamitos, CA, USA, oct 2014. IEEE
  Computer Society.

\bibitem[BEJWY22]{BenJWY22}
Omri Ben-Eliezer, Rajesh Jayaram, David~P Woodruff, and Eylon Yogev.
\newblock A framework for adversarially robust streaming algorithms.
\newblock {\em ACM Journal of the ACM (JACM)}, 69(2):1--33, 2022.

\bibitem[BK12]{BrakerskiK12}
Zvika Brakerski and Yael~Tauman Kalai.
\newblock {Efficient Interactive Coding against Adversarial Noise}.
\newblock In {\em 2012 IEEE 53rd Annual Symposium on Foundations of Computer
  Science}, pages 160--166, 2012.

\bibitem[BN13]{BrakerskiN13}
Zvika Brakerski and Moni Naor.
\newblock {Fast Algorithms for Interactive Coding}.
\newblock In {\em Proceedings of the Twenty-Fourth Annual ACM-SIAM Symposium on
  Discrete Algorithms}, SODA '13, page 443–456, USA, 2013. Society for
  Industrial and Applied Mathematics.

\bibitem[BR11]{BravermanR11}
Mark Braverman and Anup Rao.
\newblock {Towards Coding for Maximum Errors in Interactive Communication}.
\newblock In {\em Proceedings of the Forty-Third Annual ACM Symposium on Theory
  of Computing}, STOC '11, page 159–166, New York, NY, USA, 2011. Association
  for Computing Machinery.

\bibitem[Bra12]{Braverman12}
Mark Braverman.
\newblock {Towards Deterministic Tree Code Constructions}.
\newblock In {\em Proceedings of the 3rd Innovations in Theoretical Computer
  Science Conference}, ITCS '12, page 161–167, New York, NY, USA, 2012.
  Association for Computing Machinery.

\bibitem[CCFC02]{CharikarCF02}
Moses Charikar, Kevin Chen, and Martin Farach-Colton.
\newblock Finding frequent items in data streams.
\newblock In {\em International Colloquium on Automata, Languages, and
  Programming}, pages 693--703. Springer, 2002.

\bibitem[CCFC04]{CharikarCF04}
Moses Charikar, Kevin Chen, and Martin Farach-Colton.
\newblock Finding frequent items in data streams.
\newblock {\em Theoretical Computer Science}, 312(1):3--15, 2004.
\newblock Automata, Languages and Programming.

\bibitem[CM05]{CormodeM05}
Graham Cormode and S.~Muthukrishnan.
\newblock An improved data stream summary: the count-min sketch and its
  applications.
\newblock {\em Journal of Algorithms}, 55(1):58--75, 2005.

\bibitem[CZ16]{Chen16}
Jiecao Chen and Qin Zhang.
\newblock Bias-aware sketches.
\newblock {\em arXiv preprint arXiv:1610.07718}, 2016.

\bibitem[DEL{\etalchar{+}}22]{DinurELM22}
Irit Dinur, Shai Evra, Ron Livne, Alexander Lubotzky, and Shahar Mozes.
\newblock Locally testable codes with constant rate, distance, and locality.
\newblock In {\em Proceedings of the 54th Annual ACM SIGACT Symposium on Theory
  of Computing}, pages 357--374, 2022.

\bibitem[DGY11]{DvirGY11}
Zeev Dvir, Parikshit Gopalan, and Sergey Yekhanin.
\newblock Matching vector codes.
\newblock {\em SIAM Journal on Computing}, 40(4):1154--1178, 2011.

\bibitem[DHM{\etalchar{+}}15]{DaniHMSY15}
Varsha Dani, Thomas~P. Hayes, Mahnush Movahedi, Jared Saia, and Maxwell Young.
\newblock {Interactive Communication with Unknown Noise Rate}, 2015.

\bibitem[Efr09]{Efremenko09}
Klim Efremenko.
\newblock 3-query locally decodable codes of subexponential length.
\newblock In {\em Proceedings of the forty-first annual ACM symposium on Theory
  of computing}, pages 39--44, 2009.

\bibitem[EGH16]{EfremenkoGH16}
Klim Efremenko, Ran Gelles, and Bernhard Haeupler.
\newblock {Maximal Noise in Interactive Communication Over Erasure Channels and
  Channels With Feedback}.
\newblock {\em {IEEE} Trans. Inf. Theory}, 62(8):4575--4588, 2016.

\bibitem[EHK{\etalchar{+}}23]{EfremenkoHKKRS23}
Klim Efremenko, Bernhard Haeupler, Yael~Tauman Kalai, Gillat Kol, Nicolas
  Resch, and Raghuvansh~R Saxena.
\newblock Interactive coding with small memory.
\newblock In {\em Proceedings of the 2023 Annual ACM-SIAM Symposium on Discrete
  Algorithms (SODA)}, pages 3587--3613. SIAM, 2023.

\bibitem[EKS20]{EfremenkoKS20b}
Klim Efremenko, Gillat Kol, and Raghuvansh~R. Saxena.
\newblock {Binary Interactive Error Resilience Beyond ${{}^{1}}\!/\!_{8}$ (or
  why $({{}^{1}}\!/\!_{2})^{3} > {{}^{1}}\!/\!_{8})$}.
\newblock In {\em 2020 IEEE 61st Annual Symposium on Foundations of Computer
  Science (FOCS)}, pages 470--481, 2020.

\bibitem[Fla85]{Flajolet85}
Philippe Flajolet.
\newblock Approximate counting: a detailed analysis.
\newblock {\em BIT Numerical Mathematics}, 25(1):113--134, 1985.

\bibitem[For66]{Forney66}
G.~Forney.
\newblock Generalized minimum distance decoding.
\newblock {\em IEEE Transactions on Information Theory}, 12(2):125--131, 1966.

\bibitem[Gel17]{Gelles-survey}
Ran Gelles.
\newblock {Coding for Interactive Communication: A Survey}.
\newblock {\em Foundations and Trends® in Theoretical Computer Science},
  13:1--161, 01 2017.

\bibitem[GGS24]{GuptaGS24}
Meghal Gupta, Venkatesan Guruswami, and Mihir Singhal.
\newblock Tight bounds for stream decodable error-correcting codes, 2024.

\bibitem[GH13]{GhaffariH13}
Mohsen Ghaffari and Bernhard Haeupler.
\newblock {Optimal Error Rates for Interactive Coding II: Efficiency and List
  Decoding}.
\newblock {\em Proceedings - Annual IEEE Symposium on Foundations of Computer
  Science, FOCS}, 12 2013.

\bibitem[GH17]{GellesH17}
Ran Gelles and Bernhard Haeupler.
\newblock {Capacity of Interactive Communication over Erasure Channels and
  Channels with Feedback}.
\newblock {\em SIAM Journal on Computing}, 46:1449--1472, 01 2017.

\bibitem[GHK{\etalchar{+}}16]{GellesHKRW16}
Ran Gelles, Bernhard Haeupler, Gillat Kol, Noga Ron-Zewi, and Avi Wigderson.
\newblock {\em {Towards Optimal Deterministic Coding for Interactive
  Communication}}, pages 1922--1936.
\newblock 2016.

\bibitem[GI18]{GellesI18}
Ran Gelles and Siddharth Iyer.
\newblock {Interactive coding resilient to an unknown number of erasures}.
\newblock {\em arXiv preprint arXiv:1811.02527}, 2018.

\bibitem[GKLR21]{Garg21}
Sumegha Garg, Pravesh~K Kothari, Pengda Liu, and Ran Raz.
\newblock Memory-sample lower bounds for learning parity with noise.
\newblock {\em arXiv preprint arXiv:2107.02320}, 2021.

\bibitem[GS92]{GemmellS92}
Peter Gemmell and Madhu Sudan.
\newblock Highly resilient correctors for polynomials.
\newblock {\em Information Processing Letters}, 43(4):169--174, 1992.

\bibitem[GS00]{GuruswamiS00}
Venkatesan Guruswami and Madhu Sudan.
\newblock List decoding algorithms for certain concatenated codes.
\newblock In {\em Proceedings of the Thirty-Second Annual ACM Symposium on
  Theory of Computing}, STOC '00, page 181–190, New York, NY, USA, 2000.
  Association for Computing Machinery.

\bibitem[Gur06]{Guruswami06}
Venkat Guruswami.
\newblock Error-correcting codes: Constructions and algorithms, lecture no. 11,
  2006.

\bibitem[GZ22a]{GuptaZ22c}
Meghal Gupta and Rachel~Yun Zhang.
\newblock Efficient interactive coding achieving optimal error resilience over
  the binary channel.
\newblock {\em arXiv preprint arXiv:2207.01144}, 2022.

\bibitem[GZ22b]{GuptaZ22a}
Meghal Gupta and Rachel~Yun Zhang.
\newblock {The Optimal Error Resilience of Interactive Communication Over
  Binary Channels}.
\newblock In {\em Symposium on Theory of Computing, {STOC} 2012, New York, NY,
  USA, June 20 - June 24, 2022}, STOC '22. {ACM}, 2022.

\bibitem[Hae14]{Haeupler14}
Bernhard Haeupler.
\newblock {Interactive Channel Capacity Revisited}.
\newblock In {\em 55th {IEEE} Annual Symposium on Foundations of Computer
  Science, {FOCS} 2014, Philadelphia, PA, USA, October 18-21, 2014}, pages
  226--235, 2014.

\bibitem[Ham50]{Hamming50}
Richard~W Hamming.
\newblock Error detecting and error correcting codes.
\newblock {\em The Bell system technical journal}, 29(2):147--160, 1950.

\bibitem[IW05]{IndykW05}
Piotr Indyk and David Woodruff.
\newblock Optimal approximations of the frequency moments of data streams.
\newblock In {\em Proceedings of the thirty-seventh annual ACM symposium on
  Theory of computing}, pages 202--208, 2005.

\bibitem[JL84]{JohnsonL84}
William Johnson and Joram Lindenstrauss.
\newblock Extensions of lipschitz maps into a hilbert space.
\newblock {\em Contemporary Mathematics}, 26:189--206, 01 1984.

\bibitem[KSY14]{KoppartySY14}
Swastik Kopparty, Shubhangi Saraf, and Sergey Yekhanin.
\newblock High-rate codes with sublinear-time decoding.
\newblock {\em J. ACM}, 61(5), sep 2014.

\bibitem[Mor78]{Morris78}
Robert Morris.
\newblock Counting large numbers of events in small registers.
\newblock {\em Communications of the ACM}, 21(10):840--842, 1978.

\bibitem[MRU11]{McGregorRU11}
Andrew McGregor, Atri Rudra, and Steve Uurtamo.
\newblock Polynomial fitting of data streams with applications to codeword
  testing.
\newblock {\em Symposium on Theoretical Aspects of Computer Science
  (STACS2011)}, 9, 03 2011.

\bibitem[Mul54]{Muller54}
David~E Muller.
\newblock Application of boolean algebra to switching circuit design and to
  error detection.
\newblock {\em Transactions of the IRE professional group on electronic
  computers}, (3):6--12, 1954.

\bibitem[MW10]{MonemizadehW10}
Morteza Monemizadeh and David~P Woodruff.
\newblock 1-pass relative-error lp-sampling with applications.
\newblock In {\em Proceedings of the twenty-first annual ACM-SIAM symposium on
  Discrete Algorithms}, pages 1143--1160. SIAM, 2010.

\bibitem[MWOC21]{Ma21}
Chaoyi Ma, Haibo Wang, Olufemi Odegbile, and Shigang Chen.
\newblock Noise measurement and removal for data streaming algorithms with
  network applications.
\newblock In {\em 2021 IFIP Networking Conference (IFIP Networking)}, pages
  1--9. IEEE, 2021.

\bibitem[Nis90]{Nisan90}
N.~Nisan.
\newblock Pseudorandom generators for space-bounded computations.
\newblock In {\em Proceedings of the Twenty-Second Annual ACM Symposium on
  Theory of Computing}, STOC '90, page 204–212, New York, NY, USA, 1990.
  Association for Computing Machinery.

\bibitem[Ree54]{Reed54}
Irving~S Reed.
\newblock A class of multiple-error-correcting codes and the decoding scheme.
\newblock {\em IEEE Transactions on Information Theory}, 4(4):38--49, 1954.

\bibitem[RS60]{ReedS60}
I.~S. Reed and G.~Solomon.
\newblock Polynomial codes over certain finite fields.
\newblock {\em Journal of the Society for Industrial and Applied Mathematics},
  8(2):300--304, 1960.

\bibitem[RU10]{RudraU10}
Atri Rudra and Steve Uurtamo.
\newblock Data stream algorithms for codeword testing.
\newblock In {\em International Colloquium on Automata, Languages, and
  Programming}, pages 629--640. Springer, 2010.

\bibitem[Sch80]{Schwartz80}
J.~T. Schwartz.
\newblock Fast probabilistic algorithms for verification of polynomial
  identities.
\newblock {\em J. ACM}, 27(4):701–717, oct 1980.

\bibitem[Sch92]{Schulman92}
Leonard~J. Schulman.
\newblock {Communication on noisy channels: a coding theorem for computation}.
\newblock In {\em Proceedings., 33rd Annual Symposium on Foundations of
  Computer Science}, pages 724--733, 1992.

\bibitem[Sch93]{Schulman93}
Leonard~J. Schulman.
\newblock {Deterministic Coding for Interactive Communication}.
\newblock In {\em Proceedings of the Twenty-Fifth Annual ACM Symposium on
  Theory of Computing}, STOC '93, page 747–756, New York, NY, USA, 1993.
  Association for Computing Machinery.

\bibitem[Sch96]{Schulman96}
Leonard~J. Schulman.
\newblock {Coding for interactive communication}.
\newblock {\em IEEE Transactions on Information Theory}, 42(6):1745–1756,
  1996.

\bibitem[Sha48]{Shannon48}
Claude~E. Shannon.
\newblock {A mathematical theory of communication}.
\newblock {\em The Bell System Technical Journal}, 27(3):379--423, 1948.

\bibitem[Y{\etalchar{+}}12]{Yekhanin12}
Sergey Yekhanin et~al.
\newblock Locally decodable codes.
\newblock {\em Foundations and Trends{\textregistered} in Theoretical Computer
  Science}, 6(3):139--255, 2012.

\bibitem[Yek12]{Yekhanin11}
Sergey Yekhanin.
\newblock Locally decodable codes.
\newblock {\em Foundations and Trends® in Theoretical Computer Science},
  6(3):139--255, 2012.

\bibitem[Zip79]{Zippel79}
Richard Zippel.
\newblock Probabilistic algorithms for sparse polynomials.
\newblock In {\em Symbolic and Algebraic Computation}, 1979.

\end{thebibliography}

\end{document}